\documentclass[a4paper,11pt]{article}

\usepackage[english]{babel}
\usepackage{dsfont}
\usepackage{setspace}
\usepackage[a4paper,tmargin=3.5truecm,bmargin=4truecm,rmargin=2.5truecm,
lmargin=3.2truecm,twoside,verbose=true]{geometry}
\usepackage{cancel,graphicx}
\usepackage{hyperref}
\usepackage{amsmath,amssymb,slashed}

\usepackage{amsmath,amssymb}
\usepackage[amsmath, hyperref, thmmarks]{ntheorem}
\numberwithin{equation}{section}

\allowdisplaybreaks[1]

\newcommand\bbbone{\mathbb{I}}

\newcommand\del{{\partial}}
\newcommand\delbar{{{\bar{\partial}}}}

\theoremsymbol{}
\theorembodyfont{\slshape}
\theoremheaderfont{\normalfont\bfseries}
\theoremseparator{}
\newtheorem{Theorem}{Theorem}[section]
\newtheorem{theorem}[Theorem]{Theorem}

\newtheorem{proposition}[Theorem]{Proposition}

\newtheorem{lemma}[Theorem]{Lemma}

\newtheorem{corollary}[Theorem]{Corollary}
\theorembodyfont{\upshape}
\theoremsymbol{\ensuremath{\blacklozenge}}

\newtheorem{remark}[Theorem]{Remark}

\newtheorem{definition}[Theorem]{Definition}
\theoremstyle{nonumberplain}
\theoremheaderfont{\scshape}
\theorembodyfont{\normalfont}
\theoremsymbol{\ensuremath{\blacksquare}}

\newtheorem{proof}{Proof}
\qedsymbol{\ensuremath{_\blacksquare}}
\theoremclass{LaTeX}

\renewenvironment{thebibliography}[1]
         {\section*{References}\frenchspacing\small
          \begin{list}{[\arabic{enumi}]}
         {\usecounter{enumi}\parsep=2pt\topsep 0pt
         \settowidth{\labelwidth}{[#1]}
         \leftmargin=\labelwidth\advance\leftmargin\labelsep
         \rightmargin=0pt\itemsep=1pt\sloppy}}{\end{list}}
\title{Connes distance by examples: \\ Homothetic spectral metric spaces}
\author{Jean-Christophe Wallet}

\begin{document}
\date{}
\maketitle
\vspace*{-1cm}

\begin{center}
\textit{Laboratoire de Physique Th\'eorique, B\^at.\ 210\\
    Universit\'e Paris-Sud 11,  91405 Orsay Cedex, France\\
    e-mail: 
\texttt{jean-christophe.wallet@th.u-psud.fr}}\\[1ex]

\end{center}

\vskip 2cm



\begin{abstract}
We study metric properties stemming from the Connes spectral distance on three types of non compact noncommutative spaces which have received attention recently from various viewpoints in the physics literature. These are the noncommutative Moyal plane, a family of harmonic Moyal spectral triples for which the Dirac operator squares to the harmonic oscillator Hamiltonian and a family of spectral triples with Dirac operator related to the Landau operator. We show that these triples are homothetic spectral metric spaces, having an infinite number of distinct pathwise connected components. The homothetic factors linking the distances are related to determinants of effective Clifford metrics. We obtain as a by product new examples of explicit spectral distance formulas. The results are discussed.
\end{abstract}

\pagebreak
\section{Introduction}
\subsection{Motivations and organization}
In noncommutative geometry (for reviews, see e.g.~\cite{CONNES, CM, LANDI, GRACIAVAR}), the concept of spectral triple involves a natural notion of distance, stemming from the initial observation by Connes \cite{CONNES1} that Dirac operator can actually be used to generate metric data. This is known as the Connes distance \cite{CONNES1}, \cite{CONNES2}, hereafter called the spectral distance. The spectral distance may be viewed as a noncommutative analog of the geodesic distance. Indeed, in the commutative case, for a finite dimensional compact Riemann spin manifold $M$ described by the standard spectral triple built from $\mathbb{A}=C^\infty(M)$, $H$ the Hilbert space of square integrable spinors on $M$ and $D$ the usual Dirac operator, the spectral distance between pure states of $C^\infty({M})$, i.e points, coincides with the geodesic distance between those points while for non pure states, the spectral distance is the Wasserstein distance of order 1 between the corresponding probability distributions in the theory of optimal transport \cite{Rieffel1}.\par 

Many examples of noncommutative spaces have been now constructed, but comparatively little work has been done so far on their metric aspects. From an abstract viewpoint, the notion of noncommutative metric space has been developed by Rieffel in \cite{Rieffel1, Rieffel11, Rieffel2a, Rieffel2b, Rieffel3} leading to a natural theory of noncommutative compact metric spaces. An extention to the locally compact case, initiated in \cite{Latremolieres}, deserves further investigations. The computation of explicit spectral distance formulas is difficult, due to numerous technical points to be overcome, unless the (noncommutative) geometry is ''relatively simple''. Therefore, the first computations were related to a limited number of situations, namely lattice geometries \cite{Muller1}, \cite{Muller2}, \cite{LIZZI}, finite-dimensional algebras \cite{Mart1}, almost commutative geometries \cite{Mart2}, \cite{Martwulk}, showing a relationship with the Carnot-Caratheodory distance in sub-Riemannian geometry \cite{Mart3} or providing a metric interpretation of the Higgs field as the component of the metric in a discrete internal dimension \cite{Martwulk}.\par

Recently, the spectral distance has been studied \cite{WAL1}, \cite{WAL2} within noncommutative Moyal plane. This latter is described by the spectral triple proposed in \cite{GAYRAL2004, MARSE1} (hereafter called Standard Moyal Spectral Triple) as a noncommutative analog of (non compact) Riemann spin geometry. An explicit formula for the spectral distance between pure states related to the eigenfunctions of the harmonic oscillator has been obtained in \cite{WAL1, WAL2}, showing in particular that the spectral distance $d_D$ cannot be viewed as a deformation of the usual Euclidean distance on $\mathbb{R}^2$, despite the fact the Moyal plane can be interpreted as an isospectral deformation of $\mathbb{R}^2$ \cite{WAL1, WAL2}. A comparison to the quantum distance introduced by Doplicher, Fredenhagen and Roberts (DFR) \cite{DFR} has appeared recently \cite{mart-tom}. The above distance formula has been extended \cite{mart-recent} to some classes of coherent states (i.e the ``quantum points''). Obviously, enlarging the above list of explicit examples of spectral distances appears desirable, especially in the case of non compact noncommutative spaces for which a general theory supplementing the above compact case remains to be done.\par 

As far as physics is concerned, Moyal geometry or differential calculi linked with Moyal algebra(s) have received a lot of attention within noncommutative field theories (see e.g \cite{GROSSWULK0}-\cite{LSZ}) and noncommutative gauge theories (see e.g \cite{WAL6}-\cite{WAL5}). The recent constructions of the first renormalisable (bosonic or fermionic) noncommutative field theories have pointed toward other interesting spectral triples for which the Dirac operator is no longer the usual Dirac operator of the Standard Moyal Spectral Triple. For the (bosonic) renormalisable noncommutative harmonic $\varphi^4$ scalar model \cite{GROSSWULK0, GROSSEWULK}, the Dirac operator is a''square root'' of the harmonic oscillator Hamiltonian. This gives rise to an interesting class of spectral triples, first considered in \cite{GROSSEWULK2} (hereafter called Harmonic Moyal Spectral Triples). This has been further investigated in \cite{wulkgayral2010} (see also \cite{wulk-finite}), from the viewpoint of spectral dimension and spectral action computation as an attempt to understand more deeply the noncommutative structures behind the gauge invariant model derived in \cite{WAL6}, \cite{GROSSEWOHL}. For the (fermionic) renormalisable noncommutative Gross-Neveu model \cite{VIGNES1, VIGNEWALLET}, the Dirac operator is built from the Landau operator with magnetic field proportional to the deformation parameter of the $\mathbb{R}^2$-plane and can be used to built a spectral triple (hereafter called Landau Moyal Spectral Triple). It does not square to the harmonic oscillator Hamiltonian but instead is related, e.g to the kinetic operator occurring in the class of noncommutative LSZ models \cite{LSZ}.\par

The purpose of this paper is to relate the metric properties encoded in the Standard, Harmonic and Landau Spectral Triples. We find that these spectral triples describe non compact spectral metric spaces. These are homothetic to each other, with an infinite number of distinct path-connected components. As a by product, we find new examples of spectral distance formulas that extend \cite{WAL1}, \cite{WAL2}. Here, homothety is due in part to the algebraic property that the spatial derivatives can be expressed as inner derivatives on the multiplier algebra of the algebra involved in the triples (see below), $\partial_\mu f=-{{i}\over{2}}[{\tilde{x}}_\mu,f]_\star$. This implies that, for each spectral triple, the operator characterizing $\ell_D$ (see equation \eqref{lipnorm}) can always be written{\footnote{We use Einstein summation convention over repeated indices.}} as $[D,\pi(a)]=\Gamma^\mu\pi(\partial_\mu a)$, where the set of hermitian matrices $\Gamma^\mu$ spans a representation of a Clifford algebra and depends on the triple. The homothetic factors linking the distances are then related to the determinants of the effective metrics induced by these ''effective'' Clifford algebras. This is summarized in the Theorem \ref{summaryresult} of the section 4 where we also discuss the results and conclude. In the subsection \ref{standardmoyalgeo}, we recall basic features on the Standard Moyal Spectral Triple and fix the notations. To have a self-contained presentation, additional related properties used in the course of the discussion are collected in the appendix. The subsection \ref{harmonictriple} deals with the harmonic Moyal spectral triple. In the section \ref{landautriple}, we consider the Landau Moyal Spectral Triple. We show the localized compactness condition (see iv) of Definition \ref{spectraltriple} below) for the resolvent operator associated to the Dirac operator. This latter is further shown to be related to a noncommutative connection on a certain module.\\
{\bf{Note added:}} After the completion of this work, we became aware of a second version of \cite{mart-tom} proposing to use Theorem \ref{th1} (or Theorem \ref{summaryresult}) given below to define a classical limit (i.e a limit for which the Planck length tends to zero) for which the spectral distance between coherent states for the Moyal plane case reduces to the usual Euclidean distance.

\subsection{Definitions and general properties}
In this paper{\footnote{In the following, ${\cal{B}}({{H}})$ and ${\cal{K}}({{H}})$ denote the C*-algebra of bounded operators on a separable Hilbert space ${{H}}$ and its ideal (C*-subalgebra) of compact operators on ${{H}}$.}}, we will use the following definition for a spectral triple (viewed as a ''$K$-cycle''):
\begin{definition}\label{spectraltriple} 
A spectral triple is the set of data $\mathfrak{{X}}_D=(\mathbb{A},\ \pi,\ {{H}},\ D)$ in which:\\
i) $\mathbb{A}$ is an involutive algebra and $\pi$ is a faithful $\star$-representation of $\mathbb{A}$ on ${\cal{B}}({{H}})$, ii) $D$ is a self-adjoint possibly unbounded operator defined on a dense domain Dom$(D)\subset{{H}}$ with $\pi(a)$Dom$(D)$$\subset$Dom$(D)$, $\forall a\in\mathbb{A}$ and satisfying: iii) $\forall a\in\mathbb{A}$, $[D,\pi(a)]\in{\cal{B}}({{H}})$, iv) $\forall a\in\mathbb{A}$, $\forall\lambda\notin$ sp$(D)$, $\pi(a)(D-\lambda)^{-1}\in{\cal{K}}({{H}})$.
\end{definition}
The metric information for the underlying noncommutative space can be extracted from $\mathfrak{{X}}_D$. The definition \ref{spectraltriple} has to be supplemented by additional conditions in order to give rise to natural noncommutative analogs of manifolds. For more details and examples, see e.g \cite{CONNES, LANDI, GRACIAVAR}. When $\mathbb{A}$ is unital, note that iv) of definition \ref{spectraltriple} is equivalent to $(D-\lambda\bbbone_\mathbb{A})^{-1}\in{\cal{K}}({{H}})$ for any $\lambda\notin$ spec$(D)$, the spectrum of $D$. \\
Let ${\mathfrak{S}}(\mathbb{A})$ be the space of states of $\mathbb{A}$, i.e positive linear map $\omega:\mathbb{A}\to\mathbb{C}$ with norm 1.
\begin{definition}\cite{CONNES1, CONNES2}\label{spectraldistance} The spectral distance between any two states is given by:
\begin{equation}
d_D(\omega_1,\omega_2)=\sup_{a\in\mathbb{A}}\big\{ \vert \omega_1(a)-\omega_2(a)\vert;\ \ell_D(a)\le1\big \},\ \forall\omega_1,\omega_2\in{\mathfrak{S}}(\mathbb{A}) \label{specdist},
\end{equation}
in which the D-seminorm on $\mathbb{A}$, $\ell_D:\mathbb{A}\to{\mathbb{R}}^+$, associated with the Dirac operator $D$ is defined by
\begin{equation}
\ell_D:a\mapsto \ell_D(a):=||[D,\pi(a)]||,\  \forall a\in\mathbb{A}\label{lipnorm},
\end{equation}
where $\vert\vert.\vert\vert$ is the operator norm for the representation of $\mathbb{A}$ in ${\cal{B}}({{H}})$. 
\end{definition}
The definition \ref{spectraldistance} extends to the case where $\mathbb{A}$ is a pre-C* algebra (such as the algebra underlying the spectral triples considered in this paper) for which the notion of state is meaningfull. Indeed, the restriction on $\mathbb{A}$ of any state on ${\bar{\mathbb{A}}}$, the C* completion of $\mathbb{A}$, defines on $\mathbb{A}$ a unique positive linear map with norm 1 while continuity insures that any positive linear map of norm 1 on $\mathbb{A}$ determines a unique state on ${\bar{\mathbb{A}}}$.\par
In the sequel, we will use a simple natural definition for a spectral metric space:
\begin{definition}\label{spectralmetricspace}
A spectral metric space is a spectral triple $\mathfrak{{X}}_D=(\mathbb{A},\ \pi,\ {{H}},\ D)$ in the sense of Definition \ref{spectraltriple} where a): the representation $\pi$ of $\mathbb{A}$ in ${\cal{B}}({{H}})$ is non degenerate and b): the metric commutant $\mathbb{A}_D^\prime:=\{a\in\mathbb{A};\ [D,\pi(a)]=0 \}$ is trivial.
\end{definition}
This definition is a bit more restrictive than the one used in \cite{Rieffel1}-\cite{Rieffel3} which is based on the notion of order-unit spaces. Note that the definition we use here does not require the $d_D$-metric topology on $\mathfrak{S}(\mathbb{A})$ to be of a specific type fixed once in advance and forever. Let us comment this point. \par
The notion of noncommutative metric space developed in \cite{Rieffel1}-\cite{Rieffel3} exploits the initial observation of \cite{CONNES1} that equipping a C*-algebra $\mathbb{A}$ with a metric can be done by choosing a suitable densely defined seminorm $\ell$ on $\mathbb{A}$. This latter may be viewed as a generalization of a Lipschitz seminorm. Indeed, for a (commutative) compact metric space (X,$\rho$), the metric $\rho$ can be recovered from the usual Lipschitz seminorm $\ell_\rho$ on (the commutative unital algebra) $\mathbb{A}=C(X)$, given for any $f\in C(X)$ by 
\begin{equation}
\ell_\rho(f):=\sup\{{{\vert f(x)-f(y) \vert}\over{\rho(x,y)}};\ x,y\in X, x\ne y\},
\end{equation}
thanks to the following relation 
\begin{equation}
\rho(x,y)=\sup\{\vert f(x)-f(y)\vert;\  f\in C(X),\ \ell_\rho(f)\le1\}. 
\end{equation}
This extends to the Kantorovitch distance on the space of probability measures over X, given for any $\mu_1,\mu_2\in{\mathfrak{S}}(C(X))$ by 
\begin{equation}
\rho(\mu_1,\mu_2)=\sup\{\vert \mu_1(f)-\mu_2(f)\vert;\ f\in C(X),\ \ell_\rho(f)\le1\}
\end{equation}
and is known to metrize the weak* topology on ${\mathfrak{S}}(C(X))$. A natural extension to the noncommutative unital case is to determine the properties obeyed by $\ell_D$, the analog of the Lipschitz seminorm, such that the spectral distance $d_D$ induces the weak* topology on $\mathfrak{S}(\mathbb{A})$, as carried out by Rieffel in \cite{Rieffel1, Rieffel11, Rieffel2a, Rieffel2b, Rieffel3}. This produced the theory of noncommutative compact metric spaces mentionned above{\footnote{These are called (compact) quantum metric spaces by Rieffel.}}. For unital $\mathbb{A}$, provided conditions a) and b) of Definition \ref{spectralmetricspace} hold, $d_D$ metrizes the weak* topology on ${\mathfrak{S}}(\mathbb{A})$ if and only if the ``Lipschitz ball'' $B({\mathfrak{X}}_D):=\{a\in\mathbb{A};\ \ell_D(a)\le1\big \}$ is norm pre-compact in  $\mathbb{A}/\mathbb{A}_D^\prime$. This is the natural condition for a spectral triple to give rise to a compact metric space. \par 

A full generalization to the non unital case is still lacking. It is worth pointing out that having a spectral distance metrizing the weak* topology on the space of states can no longer be a systematic requirement (unless the space is ``sufficiently close to the compact case''). Indeed, for commutative locally compact metric space, the Kantorovitch distance does not in general induces the weak* topology on $\mathfrak{S}(C_0(X))$. Nevertheless, a partial extension to the non unital case has been done in \cite{Latremolieres}, adapting to the noncommutative case the fact that each distance in the set of bounded-Lipschitz distances{\footnote{The set of bounded-Lipschitz distances $(d^\alpha_\ell)_{\alpha\in\mathbb{R}^+}$ is defined by $d^\alpha_\ell(\mu_1,\mu_2):=\sup_{f\in C_0(X)}(|\mu_1(f)-\mu_2(f)|;\ \ell(f)\le1,\ ||f||\le\alpha)$, $\forall\mu_1,\mu_2\in{\mathfrak{S}}(\mathbb{A})$. $d^\infty_\ell$ coincides with the Kantorovitch distance.}} induces the weak* topology on $\mathfrak{S}(C_0(X))$ which involves the Kantorovitch distance whenever $X$ is bounded locally compact. It is found \cite{Latremolieres} that for a non unital (separable) C*-algebra (provided ${\it{a)}}$ and ${\it{b)}}$ of Definition \ref{spectralmetricspace} still hold), $d_D$ induces the weak* topology on ${\mathfrak{S}}(\mathbb{A})$ if and only if one can find a strictly positive element{\footnote{$h\in\mathbb{A}$ is strictly positive if it is positive and $h\mathbb{A}h$ is norm-dense in $\mathbb{A}$}} $h\in\mathbb{A}$ such that $hB({\mathfrak{X}}_D)h$ is norm pre-compact in ${\mathbb{A}}$. This is a natural condition for having a noncommutative {\it{bounded}} locally compact metric space. Note that spectral triples fullfilling this condition have been studied recently in \cite{bellissardmarcolli}. The spectral triples we will consider in this paper do not correspond to compact or bounded locally compact metric spaces but satisfy the conditions of Definition \ref{spectralmetricspace}. They can therefore be viewed as non compact spectral metric spaces.\par 

\section{Standard and Harmonic Moyal geometries.}\label{section2}
\subsection{The Standard Moyal Spectral Triple}\label{standardmoyalgeo}
In this subsection, we collect the material related to the Standard Moyal Spectral Triple proposed in \cite{GAYRAL2004, MARSE1}. Other useful properties are given in the appendix. Note that it can be viewed as an isospectral deformation of the canonical commutative spectral triple for the Euclidean $\mathbb{R}^2$ plane. For more details, see e.g \cite{GRACIAVAR}, \cite{Gracia-Bondia:1987kw, Varilly:1988jk}. \\
Let ${\cal{S}}= {\cal{S}}({\mathbb{R}}^2)$ be the space of complex-valued Schwartz functions on ${\mathbb{R}}^2$ and ${\cal{S}}^\prime(\mathbb{R}^2)={\cal{S}}^\prime$ its topological dual space. In the following, $fg$ denotes the commutative product of any two functions $f,g\in{\cal{S}}$ ($m$ is the multiplication operator $m(f)g:=fg$), $||.||_2$ is the $L^2(\mathbb{R}^2)$-norm. Recall that ${\cal{S}}\subset L^2(\mathbb{R}^2)$ densely.
\begin{proposition}\label{moyalproperty}\cite{Gracia-Bondia:1987kw, Varilly:1988jk}
The  associative Moyal $\star$-product is defined for all $f, g$ in ${\cal{S}}$ by: $\star:{\cal{S}}\times{\cal{S}}\to{\cal{S}}$ 
\begin{equation}
\label{moyal}
(f\star g)(x):=\frac{1}{(\pi\theta)^2}\int d^2y\,d^2z\ f(x+y)g(x+z)e^{-i\,2y^\mu\,\Theta^{-1}_{\mu\nu}z^\nu},\ \Theta_{\mu\nu}:=\theta\begin{pmatrix} 0&1 \\ -1& 0 \end{pmatrix} 
\end{equation}
where $\theta\in\mathbb{R}$, $\theta>0$. The complex conjugation is an involution $^\dag$ for the Moyal product. The integral is a faithful trace: $
\int d^2x\ (f\star g)(x)=\int d^2x\ (g\star f)(x)=\int d^2x\ f(x)g(x)$. The Leibniz rule holds: $\partial_\mu(f\star g)=\partial_\mu f\star g+f\star\partial_\mu g$, $\forall f,g\in{\cal{S}}$.
\end{proposition}
We set ${\cal{A}}:=({\cal{S}},\star)$. The $\star$-product \eqref{moyal} can be extended to spaces larger than ${\cal{S}}$, as recalled in the appendix. \par 
In this paper, we focus on metric properties so that it is sufficient to deal with the data defining a spectral triple as in Definition \ref{spectraltriple}. Recall that more conditions are needed to obtain a noncommutative analog of a non compact Riemann spin geometry. The proposal \cite{GAYRAL2004} is based on the existence of a preferred unitalisation of the initial non unital algebra. This permits one in particular to obtain an orientability condition through the construction of a suitable Hochschild cycle. The preferred unitalisation \cite{GAYRAL2004} of ${\cal{A}}$ is $({\cal{B}},\star)$, a unital Frechet pre-C* algebra, where ${\cal{B}}$ is the set of smooth bounded functions of $\mathbb{R}^2$ with all derivatives bounded. Note that one has ${\cal{A}}\subset{\bar{\cal{B}}}\subset{\mathbb{A}}_\theta\sim{\cal{L}}(L^2(\mathbb{R}))$ where ${\bar{\cal{B}}}$ is the C* completion of (${\cal{B}},\star)$ with respect to the norm $||a||:=\sup\{||a\star b||_2/||b||_2;\ b\in L^2(\mathbb{R}^2)\}$ and $\mathbb{A}_\theta$ is the unital C* algebra defined by $\mathbb{A}_\theta:=\{a\in{\cal{S}}^\prime;\ a\star b\in L^2(\mathbb{R}^2),\ \forall b\in L^2(\mathbb{R}^2) \}$. We will come back to this unitalisation in the remark \ref{unitalisation-1} below.\\
These algebras are related to the maximal unitalisation of ${\cal{A}}$, the multiplier algebra of ${\cal{A}}$ 
\begin{equation}
{\cal{M}}=\{a\in{\cal{S}}^\prime;\ a\star b\in{\cal{S}},\ b\star a\in{\cal{S}},\ \forall b\in{\cal{S}} \}\label{multiplikator},
\end{equation}
which however cannot be used here because it cannot be represented on the algebra of bounded operators of the Hilbert space of the Moyal triple ${\cal{H}}_0=L^2(\mathbb{R}^2)\otimes\mathbb{C}^2$. Note that what is (mostly) called Moyal algebra in the physics literature is the multiplier algebra ${\cal{M}}$.\par 
The next proposition \ref{formulas} collects useful formulas that we will be needed in the sequel.
\begin{proposition} \label{formulas}
We set $[f,g]_\star:=f\star g-g\star f$ and ${\tilde{x}}_\mu:=2\Theta^{-1}_{\mu\nu}x^\nu$. For any $f,g\in{\cal{M}}$, the following relations hold true:
\begin{equation}
\partial_\mu(f\star g)=\partial_\mu f\star g+f\star\partial_\mu g;\ (f\star g)^\dag=g^\dag\star f^\dag;\ \partial_\mu f=-{{i}\over{2}}[{\tilde{x}}_\mu,f]_\star\label{relat1},
\end{equation}
\begin{equation}
{\tilde{x}}_\mu\star f=({\tilde{x}}_\mu.f)+i\partial_\mu f;\ \{{\tilde{x}}_\mu,f\}_\star:={\tilde{x}}_\mu\star f+f\star{\tilde{x}}_\mu=2{\tilde{x}}_\mu.f \label{relat122}
\end{equation}
\end{proposition}
\begin{proof}
These relations can be obtained by simple calculations.
\end{proof}
Set $\partial:={{1}\over{{\sqrt{2}}}}(\partial_1-i\partial_2)$, $\delbar:={{1}\over{{\sqrt{2}}}}(\partial_1+i\partial_2)$. The standard self-adjoint Dirac operator on  $\mathbb{R}^2$ is
\begin{equation}
D_0:=-i\sigma^\mu\del_\mu = -i{\sqrt{2}}\begin{pmatrix} 0&\delbar \\ \del& 0 \end{pmatrix},\ \sigma^1=\begin{pmatrix} 0&1 \\ 1& 0 \end{pmatrix},\ 
\sigma^2=\begin{pmatrix} 0&i \\ -i& 0 \end{pmatrix},\label{cliff2}
\end{equation}
where the $\sigma^\mu$'s span an irreducible representation of the complex Clifford algebra $\mathbb{C}l_{\mathbb{C}}(2)$, $\sigma^\mu\sigma^\nu+\sigma^\nu\sigma^\mu=2\delta^{\mu\nu}$. We set $\sigma^3:=i\sigma^1\sigma^2$ and one can check $\sigma^1\sigma^3=i\sigma^2$, $\sigma^2\sigma^3=-i\sigma^1$, $\{\sigma^\mu,\sigma^3\}=0$. Let ${\cal{H}}_0= L^2({\mathbb{R}}^2)\otimes {\mathbb{C}}^2$ be the Hilbert space of square integrable sections of the trivial spinor bundle ${\mathbb{S}}={\mathbb{R}}^2\times{\mathbb{C}}^2$. The corresponding inner product is
\begin{equation}
\langle \psi,\phi\rangle=\int d^2x(\psi_1^*\phi_1+\psi_2^*\phi_2) \quad\forall
\;
\psi = \binom{ \psi_1 }{ \psi_2 }
,\;\phi = \binom{ \phi_1 }{ \phi_2 }\; \in{\cal{H}}_0.
\end{equation}
The domain of $D_0$ is Dom$(D_0)={\cal{D}}_{L^2}\otimes \mathbb{C}^2$ where ${\cal{D}}_{L^2}$ is the set of smooth functions in $L^2(\mathbb{R}^2)$ with all their derivatives in $L^2(\mathbb{R}^2)$.\par 

In the following, ${{L}}(a)\in{\cal{B}}({\cal{H}}_0)$ is the Left multiplication operator by any element of ${\cal{A}}$, defined by ${{L}}(a)\psi:=a\star\psi$, $\forall a\in{\cal{A}}$, $\forall\psi\in L^2(\mathbb{R}^2)$. One has $L(a)^\dag=L(a^\star)$. We denote by $\pi_0:{\cal{A}}\to{\cal{B}}({\cal{H}}_0)$, the faithful left regular representation of ${\cal{A}}$ on ${\cal{B}}({\cal{H}}_0)$, namely:
\begin{equation}
\pi_0(a):={{L}}(a)\otimes\bbbone_2,\ \pi_0(a)\psi=(a\star\psi_1,a\star\psi_2),\ \forall a\in{\cal{A}},\ \forall\psi=(\psi_1,\psi_2)\in{\cal{H}}_0\label{regulrep}.
\end{equation}
Let $\ell_{D_0}(a):=||[D_0,\pi_0(a)] ||$, $\forall a\in{\cal{A}}$, be the Lipschitz seminorm on ${\cal{A}}$ for $D_0$ \eqref{cliff2}.
\begin{proposition}\label{lzero}
$\ell_{D_0}(a)={\sqrt{2}}\max(||{{L}}(\partial a)||,\ ||{{L}}(\delbar a)||), \forall a\in{\cal{A}}$.
\end{proposition}
\begin{proof}
By Leibnitz rule, Proposition \ref{moyalproperty}, one infers $[\partial_\mu,{{L}}(a)]={{L}}(\partial_\mu(a))$ and therefore $[D_0,\pi(a)]=-i{{L}}(\partial_\mu  a)\otimes\sigma^\mu$ for any $a\in{\cal{A}}$. Then, by applying the general property of the operator norm $||T||^2=||TT^\dagger||$ to the operator $T=[D_0,\pi(a)]$, a direct calculation yields the result.
\end{proof}
\begin{proposition}\label{X0}\cite{GAYRAL2004}
${\mathfrak{X}}_{D_0}=({\cal{A}},\ \pi_0,\ {\cal{H}}_0=L^2({\mathbb{R}}^2)\otimes {\mathbb{C}}^2,\ D_0)$ where $\pi_0:{\cal{A}}\to{\cal{B}}({\cal{H}}_0)$ is the left regular representation \ref{regulrep} and $D_0$ is the standard Dirac operator on $\mathbb{R}^2$ \eqref{cliff2} is a spectral triple as in Definition\ref{spectraltriple}. It is the standard Moyal spectral triple.
\end{proposition}
The initial proof can be found in \cite{GAYRAL2004}. To have a self-contained presentation, we give a proof of Proposition \ref{X0} in the appendix.\par

In this paper, we will not need to deal with explicit expressions for states. A characterization of the space of pure states taken from Proposition 4 of \cite{WAL1} is given in the appendix. The explicit formula for the spectral distance between pure states related to the eigenfunctions of the harmonic oscillator constructed in \cite{WAL1, WAL2} is also given. In \cite{WAL1, WAL2}, the existence of states at infinite distance has been shown together with the fact that the Lipschitz ball for the Dirac operator is not norm bounded (see e.g section 3.5 of \cite{WAL2}). Therefore, the topology induced by $d_{D_0}$ on ${\mathfrak{S}}({\cal{A}})$ is not the weak* topology. It follows that the Standard Moyal Spectral Triple $\mathfrak{X}_{D0}$ is neither a compact nor a bounded locally compact spectral metric space.\par

\subsection{The Harmonic Moyal Spectral Triple}\label{harmonictriple}

Starting from (a $\mathbb{R}^{2n}$ version of) the spectral triple $\mathfrak{X}_{D_0}${\footnote{with additional conditions together with a preferred unitalisation described in subsection \ref{standardmoyalgeo}}}, the corresponding spectral action has been computed in \cite{GAYRAL2004}. It has been shown to be the simplest generalization of the Yang-Mills action on Moyal space which however is not renormalisable (for $n\ge2)$. Keeping in mind the renormalisability proof of a 4-dimensional $\varphi^4$ scalar theory on Moyal space given in \cite{GROSSEWULK} for which the presence of a harmonic oscillator term is essential, another interesting triple has been considered in \cite{GROSSEWULK2}. It has been further investigated in \cite{wulkgayral2010} (see also \cite{wulk-finite}), from the viewpoint of spectral dimension and computation of the spectral action. This provides a nice attempt to understand more deeply the actual noncommutative structures related to the gauge invariant model derived in \cite{WAL6}, \cite{GROSSEWOHL}. \par 
As we are presently interested by metric properties, we consider only the $K$-cycle part of this spectral triple as given in Definition \ref{spectraltriple}, called hereafter the Harmonic Moyal Spectral Triple. Its salient features is that the related Dirac operator is required to be a ''square root'' of the harmonic oscillator Hamiltonian. In the following, we assume $n=1$. The extension to any integer value for $n$ is straightforward.\par 

The construction of square roots of operators is somewhat standard in physics. In the present case, it can be conveniently carried out as follows. Introduce $(\gamma^\mu,\gamma^{\mu+2})_{\mu=1,2}$, hermitian elements of the matrix algebra $\mathbb{M}_4(\mathbb{C})$, assumed to span an irreducible representation of the complex Clifford algebra $\mathbb{C}l_{\mathbb{C}}(4)$:
\begin{equation}
\{\gamma^\mu,\gamma^\nu\}=2\delta^{\mu\nu},\ \{\gamma^{\mu+2},\gamma^{\nu+2}\}=2\delta^{\mu\nu},\ \{\gamma^\mu,\gamma^{\nu+2}\}=0,\ \mu,\nu\in\{1,2\}\label{cliff4}.
\end{equation}
Consider the following family of unbounded self adjoint Dirac operators indexed by a real parameter chosen in the range $\Omega\in ]0,1]$ to make contact with the physics literature:
\begin{equation}
D_\Omega:=\gamma^\mu(-i\partial_\mu)-\Omega\gamma^{\mu+2}m({\tilde{x}}_\mu), \label{Domega}
\end{equation}
where  $m({\tilde{x}}_\mu)a:={\tilde{x}}_\mu a$ for any $a\in{\cal{S}}$ and the domain Dom$(D_\Omega)$ is chosen to be a dense subset of  ${\cal{H}}={\cal{H}}_0\otimes \mathbb{C}^2\cong L^2(\mathbb{R}^2)\otimes{\mathbb{C}}^4$, the Hilbert space of square integrable sections of the trivial spinor bundle ${\mathbb{S}}={\mathbb{R}}^2\times{\mathbb{C}}^4$. We pick Dom$(D_\Omega)={\cal{S}}\otimes\mathbb{C}^4$. \par 

Important features for the metric properties are related to the operator $[D_\Omega,\pi(a)]$. Here, the faithful representation $\pi:{\cal{A}}\to{\cal{H}}$ denotes the extension of the left regular representation $\pi_0:{\cal{A}}\to{\cal{B}}({\cal{H}}_0)$ of Proposition \ref{X0}, namely 
\begin{equation}
\pi(a):=L(a)\otimes\bbbone_4,\ \forall a\in{\cal{A}}\label{pi4}.
\end{equation}
The Dirac operator $D_\Omega$ satisfies useful algebraic relations, stemming from \eqref{cliff4} and the properties of the Moyal product given in Proposition \ref{formulas}.
\begin{proposition} \label{algebrorelat}
The following properties hold true:
\begin{equation}
D^2_\Omega=(-\partial^2+\Omega^2{\tilde{x}}^2)\bbbone_4+i2\Omega\gamma^\mu\gamma^{\nu+2}\Theta^{-1}_{\nu\mu},\label{hamiltonian}
\end{equation}
\begin{equation}
[D_\Omega,\pi(a)]=-iL(\partial_\mu a)\otimes(\gamma^\mu+\Omega\gamma^{\mu+2}),\ \forall a\in{\cal{A}} \label{da}.
\end{equation}
\end{proposition}
\begin{proof}
For \eqref{hamiltonian}, a direct calculation, using \eqref{cliff4} leads to the result. To obtain \eqref{da}, a calculation yields
\begin{equation}
 [D_\Omega,\pi(a)]\Psi=-i\gamma^\mu\partial_\mu a\star\Psi-\Omega\gamma^{\mu+2}({\tilde{x}}_\mu(a\star\Psi)-a\star({\tilde{x}}_\mu\Psi))  
\label{interm1}.
\end{equation}
for any $\Psi$ in Dom$(D_\Omega)$ and any $a\in{\cal{A}}$. In \eqref{interm1}, recall that the ordinary multiplication of functions and the Moyal product $\star$ act componentwise on $\Psi$. Then, by combining the second term of \eqref{interm1} to the relation \eqref{relat122} and further using the fact that $\partial_\mu$ can be expressed as an inner derivation through the last relation in \eqref{relat1}, some algebraic manipulations give rise to the result \eqref{da}.
\end{proof}
It is now convenient to introduce explicit representations for the $\gamma$ matrices and focus on the corresponding Dirac operators. We therefore set 
\begin{equation}
\gamma^\mu:=\Gamma^1\otimes\sigma^\mu,\ \gamma^{\mu+2}:=\sigma^\mu\otimes\Gamma^2,\ \mu=1,2, 
\end{equation}
where the $\sigma$'s are defined in \eqref{cliff2} and $\Gamma^1$, $\Gamma^2$ are hermitian elements of $\mathbb{M}_2(\mathbb{C})$. Then, it can be verified that each of the two families of self-adjoint Dirac operators defined by
\begin{equation}
D_{1}:=(\bbbone_2\otimes\sigma^\mu)(-i\partial_\mu)-\Omega(\sigma^\mu\otimes\sigma^3)m({\tilde{x}}_\mu)\label{Domega1}
\end{equation}
\begin{equation}
D_{2}:=(\sigma^3\otimes\sigma^\mu)(-i\partial_\mu)-\Omega(\sigma^\mu\otimes\bbbone_2)m({\tilde{x}}_\mu)\label{Domega2}
\end{equation}
satisfies a relation similar to \eqref{hamiltonian}. Namely, \eqref{hamiltonian} takes the form
\begin{equation}
D^2_{1}=D^2_{2}=(-\partial^2+\Omega^2{\tilde{x}}^2)\bbbone_4-{{2\Omega}\over{\theta}}(\sigma^\mu\otimes\sigma^\mu),\label{hamiltonian12}
\end{equation}
while \eqref{da} still holds with $\gamma^\mu=\Gamma^1\otimes\sigma^\mu$, $\gamma^{\mu+2}=\sigma^\mu\otimes\Gamma^2$, $\mu=1,2$ which can be read off from \eqref{Domega1}, \eqref{Domega2}.\par

Let $d_{D_k}$, $k=1,2$ denotes the spectral distance for the Dirac operators \eqref{Domega1}, \eqref{Domega2}. In view of eqn.\eqref{specdist}, the relation \eqref{da} readily signals the existence of a simple relationship between $d_{D_k}$ and $d_{D_0}$. This can be summarized into the following statement.
\begin{theorem}\label{th1}
For any $\Omega\in]0,1]$, the triples ${\mathfrak{X}}(k):=({\cal{A}},\ \pi,\ {\cal{H}},\ D_k)$, $k=1,2$, where $\pi:{\cal{A}}\to{\cal{B}}({\cal{H}})$, $\pi(a)=L(a)\otimes\bbbone_4$, $\forall a\in{\cal{A}}$,  is the left regular representation, are spectral triples with corresponding spectral distances $d_{D_k}$ homothetic to $d_{D_0}$, namely
\begin{equation}
d_{D_k}(\omega_1,\omega_2)=(1+\Omega^2)^{-{{1}\over{2}}}d_{D_0}(\omega_1,\omega_2), \ \forall k=1,2,\ \forall\omega_1,\omega_2\in{\mathfrak{S}}({\cal{A}})\label{distanceD}.
\end{equation}
\end{theorem}
\begin{proof}
First, it is easy to check that the triples $\mathfrak{X}(k)$, $k=1,2$, fullfill the axioms i) and ii) of Definition \ref{spectraltriple}. The arguments are similar to those used in the corresponding part of the proof of Proposition \ref{X0} given in the appendix.\par
In order to deal with axiom iii) together with relation \eqref{distanceD}, we make use of the following lemma.
\begin{lemma}
For any $a\in{\cal{A}}$ and any $k=1,2$, the following relation holds true
\begin{equation}
\ell_{D_k}(a)=(1+\Omega^2)^{{{1}\over{2}}}\ell_{D_0}(a)\label{relatlipnorm}.
\end{equation}
\end{lemma}
\begin{proof}
Consider first $k=1$. The analysis is similar for $k=2$.\par

For any $a\in{\cal{A}}$, $\ell_{D_1}(a)=||[D_1,\pi(a)]||=||-iL(\partial_\mu a)\otimes(\bbbone_2\otimes\sigma^\mu+\Omega\sigma^\mu\otimes\sigma^3) ||$. It is convenient to use the explicit matrix expression given by
\begin{equation}
[D_1,\pi(a)]=-i{\sqrt{2}}\begin{pmatrix} 0&L(\delbar a)&\Omega L(\delbar a)&0\\ 
L(\partial a)&0&0&-\Omega L(\delbar a)\\
\Omega L(\partial a) &0&0&L(\delbar a)\\
0&-\Omega L(\partial a)&L(\partial a)&0
             \end{pmatrix}, \forall a\in{\cal{A}}\label{matd_1a}.
\end{equation}
In order to make the notations more compact, we set for a while $L:=L(\partial a)$, ${\bar{L}}:=L(\delbar a)$. Then, one obtains
\begin{equation}
[D_1,\pi(a)]^*[D_1,\pi(a)]=
2\begin{pmatrix}
(1+\Omega^2)L^*L&0&0&0\\
0&{\bar{L}}^*{\bar{L}}+\Omega^2L^*L&\Omega({\bar{L}}^*{\bar{L}}-L^*L)&0\\
0&\Omega({\bar{L}}^*{\bar{L}}-L^*L)&L^*L+\Omega^2{\bar{L}}^*{\bar{L}}&0\\
0&0&0&(1+\Omega^2){\bar{L}}^*{\bar{L}}
\end{pmatrix}.
\end{equation}
But, one can write
\begin{equation}
[D_1,\pi(a)]^*[D_1,\pi(a)]=(1+\Omega^2)\mathfrak{U}^\dag\begin{pmatrix}
L^*L&0&0&0\\
0&{\bar{L}}^*{\bar{L}}&0&0\\
0&0&L^*L&0\\
0&0&0&{\bar{L}}^*{\bar{L}}\end{pmatrix}\mathfrak{U}\label{diag1}
\end{equation}
with unitary $\mathfrak{U}$ given by
\begin{equation}
\mathfrak{U}=\begin{pmatrix}
1&0&0&0\\
0&(1+\Omega^2)^{-{{1}\over{2}}}&\Omega(1+\Omega^2)^{-{{1}\over{2}}}&0\\
0&-\Omega(1+\Omega^2)^{-{{1}\over{2}}}&(1+\Omega^2)^{-{{1}\over{2}}}&0\\
0&0&0&1\end{pmatrix}\label{diagprim1}.
\end{equation}
Then, \eqref{diag1}, \eqref{diagprim1} implies that for any $a\in{\cal{A}}$, 
$\ell_{D_1}(a)^2=||[D_1,\pi(a)] ||^2=
||[D_1,\pi(a)]^*[D_1,\pi(a)] ||=(1+\Omega^2)\max(||L^*L ||,||{\bar{L}}^*{\bar{L}} ||)=(1+\Omega^2)\max(||L ||^2,||{\bar{L}} ||^2)=(1+\Omega^2)\ell_{D_0}^2(a)$. This proves \eqref{relatlipnorm} and the Lemma.
\end{proof}
The above lemma implies immediately that for any $a\in{\cal{A}}$ and any $k=1,2$, $[D_k,\pi(a)]$ is a bounded operator on ${\cal{H}}$ so that the triples $\mathfrak{X}(k)$ fulfill the axiom iii) while one can write for any states $\omega_1,\omega_2\in\mathfrak{S}({\cal{A}})$
\begin{align}
d_{D_k}(\omega_1,\omega_2)&=\sup\big\{\vert \omega_1(a)-\omega_2(a)\vert;\ a\in{\cal{A}},\ell_{D_k}(a)\le1\big\}\nonumber\\
&=\sup\big\{ {{\vert \omega_1(a)-\omega_2(a)\vert}\over{\ell_{D_k}(a)}};\ a\in{\cal{A}}\big\}\nonumber\\
&=\sup\big\{ {{\vert \omega_1(a)-\omega_2(a)\vert}\over{(1+\Omega^2)^{{{1}\over{2}}}\ell_{D_0}(a)}};\ a\in{\cal{A}}\big\},
\end{align}
where the last equality stems from \eqref{relatlipnorm}, from which follows \eqref{distanceD}. \par
Finally, a standard property of the harmonic oscillator is that the spectrum of the corresponding Hamiltonian operator $H_h:=-\partial^2+\Omega^2{\tilde{x}}^2$ is spec$(H_h)={{\Omega}\over{\theta}}(n+1), n\in\mathbb{N}$ while each eigenvalue has finite multiplicity. This implies that $(-\partial^2+\Omega^2{\tilde{x}}^2)^{-1}$ is a compact operator on $L^2(\mathbb{R}^2)$, so that, in view of \eqref{hamiltonian12}, $(D_k^2+1)^{-1}$, $k=1,2$,  is also a compact operator on ${\cal{H}}$. This implies that any of the triples $\mathfrak{X}(k)$, $k=1,2$ fullfills the axiom iv) and is therefore a spectral triple. This terminates the proof of Theorem \ref{th1}.
\end{proof}
The Theorem \ref{th1} has some consequences on properties of $d_{D_k}$ and on the actual metric topological properties of the noncommutative space modeled by the harmonic spectral triple that we now discuss. \par 
First, the explicit formula obtained in \cite{WAL1,WAL2} for $\mathfrak{X}_0$ expressing the distance $d_{D_0}$ between two arbitrary pure states related to the eigenfunctions of the harmonic oscillator extends straightforwardly to the case of $\mathfrak{X}(k),\ k=1,2$. From Proposition \ref{purestates} of the appendix, these pure states are given for any $a\in{\cal{A}}$ by $\omega_m(a)={{1}\over{2\pi\theta}}\langle f_{m0},L(a)f_{m0}\rangle$, $\forall m\in\mathbb{N}$. By combining Proposition \ref{distance-basic} and \eqref{distanceD}, one obtains
\begin{equation}
d_{D_k}(\omega_m,\omega_n)={\sqrt{{{\theta}\over{2(1+\Omega^2)}} }}\sum_{p=n+1}^{m}{{1}\over{{\sqrt{p}} }},\ n<m,\ \forall k=1,2\label{formula2}.
\end{equation}
A confrontation of $d_{D_0}$ with the notion of quantum distance introduced by Doplicher, Fredenhagen and Roberts (DFR) \cite{DFR} has been performed in \cite{mart-tom}. From Theorem \ref{th1}, it follows that the conclusions of \cite{mart-tom} apply to the Harmonic Moyal Spectral Triple. In particular, the larger the energy gap between the eigenstates is (i.e $m\gg n$ in \eqref{formula2}), the closer the spectral distance $d_{D_k},\ k=1,2$ and the DFR distance are. Note that another explicit spectral distance formula among some of the coherent states (i.e the ``quantum points'') for the standard Moyal plane has been derived in \cite{mart-recent}. From this latter work, combined with Theorem \ref{th1}, one readily obtains that for the Harmonic Moyal Spectral Triple, the spectral distance $d_{D_k}, k=1,2$ between two arbitrary coherent states of the harmonic oscillator is proportional to the Euclidean distance.\par 
Next, the Theorem \ref{th1} implies that $\mathfrak{X}(k)$, $k=1,2$ is not a compact or a bounded locally compact spectral metric space, since \eqref{distanceD} holds true and $d_{D_0}$ does not metrizes the weak* topology on $\mathfrak{S}({\cal{A}})$ \cite{WAL1,WAL2}. This stems from the existence of a family of pure states defined by the family of unit vector of $L^2(\mathbb{R}^2)$ given \cite{WAL2} by
\begin{equation}
\psi_s:={{1}\over{{\sqrt{2\pi\theta}}}}\sum_{m\in\mathbb{N}}({{1}\over{\zeta(s)(m+1)^s}})^{{{1}\over{2}}}f_{m0},\ \forall 
s\in\mathbb{R},\ s>1 \label{statesinfty}
\end{equation}
where $\zeta(s)$ is the Riemann zeta function and $(f_{m0})_{m\in\mathbb{N}}$ is a subfamily of the so called matrix base (see e.g 
\cite{Gracia-Bondia:1987kw, Varilly:1988jk}) mentionned in the appendix. \par
The $d_{D_k}$-metric topological properties can be summarized as:
\begin{proposition}\label{connected1}
For any $k=1,2$, the spectral triple $\mathfrak{X}(k)$ is a spectral metric space with an infinite number of distinct connected components, each component being pathwise connected for the $d_{D_k}$-topology.
\end{proposition}
\begin{proof}
By combining the proposition 7 of \cite{WAL1} with Theorem \ref{th1} and \eqref{distanceD}, one obtains that for any $k=1,2$
\begin{equation}
d_{D_k}(\psi_{s_1},\psi_{s_2})=+\infty,\ \forall s_1,s_2\in]1,{{5}\over{4}}[\cup],{{5}\over{4}},{{3}\over{2}}],
\end{equation}
where the symbol $\psi_s$ denote the vector state generated by $\psi_{s_1}$ \eqref{statesinfty}. Then, for any $s_1$ in the above set, the subset $\mathfrak{S}_{\psi_{s_1}}\subset\mathfrak{S}({\cal{A}})$ defined by  $\mathfrak{S}_{\psi_{s_1}}:=\{\omega\in\mathfrak{S}({\cal{A}}),\ d_{D_k}(\omega,\ \psi_{s_1})<+\infty\}$ is a closed-open set for the $d_{D_k}$-topology. Indeed, for any element $\eta$ in $\mathfrak{B}_\rho(\omega)\subset\mathfrak{S}({\cal{A}})$, the open ball with center $\omega$ and radius $\rho>0$, and for any $\omega\in\mathfrak{S}_{\psi_{s_1}}$, one has 
\begin{equation}
d_{D_k}(\eta,\psi_{s_1})\le d_{D_k}(\eta,\omega)+d_{D_k}(\omega,\psi_{s_1})<+\infty, 
\end{equation}
so that $\mathfrak{S}_{\psi_{s_1}}$ is open. Now, for any $\omega$ in the complement of $\mathfrak{S}_{\psi_{s_1}}$ and any $\eta\in\mathfrak{B}_\rho(\omega)$, one has 
\begin{equation}
d_{D_k}(\omega,\psi_{s_1})\le d_{D_k}(\eta,\psi_{s_1})+d_{D_k}(\eta,\omega)
\end{equation}
implying $d_{D_k}(\eta,\psi_{s_1})=+\infty$ from which $\eta\in\mathfrak{S}_{\psi_{s_1}}$ so that $\mathfrak{S}_{\psi_{s_1}}$ is also closed. Hence, $\mathfrak{S}_{\psi_{s_1}}$ is a closed-open subset of $\mathfrak{S}({\cal{A}})$ which is therefore an union of connected parts while $\psi_s\notin\mathfrak{S}_{\psi_{s_1}}$, $\forall s\ne s_1$ and cannot belong to the same connected component.\\
Pathwise connectedness follows directly from the fact that the map $\omega_t:\ t\in[0,1]\to\mathfrak{S}({\cal{A}})$ defined by $\omega_t:=(1-t)\omega_1+t\omega_2$, for any $\omega_1,\omega_2\in\mathfrak{S}_{\psi_{s_1}}$ is $d_{D_k}$-continuous, since one has $d_{D_k}(\omega_{t_1},\omega_{t_2})=|t_1-t_2|d_{D_k}(\omega_1,\omega_2)$ from e.g \cite{Andreamartinetti} and $d_{D_k}(\omega_1,\omega_2)<+\infty$. Hence the result.
\end{proof}

\begin{remark}\label{unitalisation-1}
Replacing ${\cal{A}}$ by a suitable unitalisation in the triples $\mathfrak{X}(k)$ would not produce compact spectral metric spaces. Indeed, take the preferred unitalisation of ${\cal{A}}$, ${\cal{A}}_1=({\cal{B}},\star)$ where ${\cal{B}}$ has been defined in the subsection \ref{standardmoyalgeo}. Then, the states \eqref{statesinfty} are still states of ${\cal{A}}_1$ while ${\cal{A}}\subset{\cal{A}}_1$ implies $d_{D_k}\le d_{{\cal{A}}_1}$ where $d_{{\cal{A}}_1}$ is the spectral distance for the triple built from ${\cal{A}}_1$. \par 
\end{remark}

\section{The Landau Moyal spectral triple}\label{landautriple}

It is interesting to consider another physically motivated triple. It is built from the Dirac operator used in the renormalisable 2-D noncommutative Gross-Neveu action. Its square is related to the class of operators occurring in the action describing the Langman-Szabo-Zarembo (LSZ) model which both received some attention recently. See for instance \cite{VIGNES1}, \cite{VIGNEWALLET}, \cite{LSZ}. \par 
Let us define the following triple 
\begin{equation}
\mathfrak{Y}(\xi):=({\cal{A}},\ \pi_0,\ {\cal{H}}_0=L^2(\mathbb{R}^2)\otimes\mathbb{C}^2,\ {\cal{D}}_\xi)\label{grosstriple},
\end{equation}
in which the self-adjoint Dirac operator ${\cal{D}}_\xi$ is given by
\begin{equation}
{\cal{D}}_\xi:=-i\sigma^\mu\partial_\mu+\xi\sigma^\mu m({\tilde{x}}_\mu)\label{grossdirac}
\end{equation}
where $\pi_0$ is still defined by \eqref{regulrep} and we take Dom$({\cal{D}}_\xi)={\cal{S}}\otimes\mathbb{C}^2$. For the moment we assume $\xi\in\mathbb{R}$. Some restrictions will be put on this parameter in a while.\par 
For any $a\in{\cal{A}}$, we set $\ell_{{\cal{D}}_\xi}(a):=||[{\cal{D}}_\xi,\pi_0(a)] ||$, the Lipschitz semi-norm for ${\cal{D}}_\xi$. It appears that this latter can be simply related to the Lipschitz semi-norm of the standard Dirac operator $\ell_{D_0}$.
\begin{proposition}
The following properties hold for any $\xi\in\mathbb{R}$:
\begin{equation}
{\cal{D}}_\xi^2=(-\partial^2+\xi^2{{\tilde{x}}^2}-i2\xi{\tilde{x}}_\mu\partial_\mu)\otimes\bbbone_2-{{4\xi}\over{\theta}}\sigma^3
\label{magnethamilton},
\end{equation}
\begin{equation}
[{\cal{D}}_\xi,\pi_0(a)]=(1-\xi)L(-i\partial_\mu a)\otimes\sigma^\mu,\ \ell_{{\cal{D}}_\xi}(a)=|1-\xi|\ell_{D_0}(a), \forall a\in{\cal{A}}\label{lipnormgrossneveu}.
\end{equation}
\end{proposition}
\begin{proof}
The relation \eqref{magnethamilton} stems from a direct calculation. The first relation \eqref{lipnormgrossneveu} can be obtained, as for Proposition \ref{algebrorelat}, from a direct computation of $[{\cal{D}}_\xi,\pi_0(a)]$ and further using the relation $\{{\tilde{x}}_\mu,f\}_\star:={\tilde{x}}_\mu\star f+f\star{\tilde{x}}_\mu=2{\tilde{x}}_\mu.f$, $\forall f\in{\cal{M}}$ and the fact that $\partial_\mu$ can be actually expressed as an inner derivation through the last relation in \eqref{relat1}. From this, follows the second relation \eqref{lipnormgrossneveu}.
\end{proof}
\begin{remark} \label{remarklandau} At this point, two comments are in order:
\begin{itemize}
\item i) As expected, the Dirac operator ${\cal{D}}_\xi$ \eqref{grossdirac} does not square to the harmonic oscillator Hamiltonian. Instead, the first term in \eqref{magnethamilton} is simply the Landau Hamiltonian of the physics literature. It describes the motion of a charged particle in the $\mathbb{R}^2$-plane submitted to an external constant magnetic field. It is defined by:
\begin{equation}
H_L:=-\partial^2+\xi^2{\tilde{x}}^2-i2\xi{\tilde{x}}_\mu\partial_\mu=P_\mu^2\ge0,
\end{equation}
where $P_\mu$ is the Landau operator defined by $P_\mu:=-i\partial_\mu-B_{\mu\nu}x^\nu$ in which the skew symmetric ${\mathbb{R}}$-valued matrix $B_{\mu\nu}$ encoding the magnetic field is given by $B_{\mu\nu}=-2\xi\Theta^{-1}_{\mu\nu}$. \\
Setting $B:=B_{12}={{2\xi}\over{\theta}}$ and assuming $B>0$, one recovers the main situation considered in \cite{LSZ} where in particular the positive operator $(-D^2)$ (in the notations of \cite{LSZ}) is $(-D^2)=H_L$. The spectrum of $H_L$ is spec$(H_L)={{8\xi}\over{\theta}}(n+{{1}\over{2}})$, $n\in\mathbb{N}$, related to the energies of the Landau levels, with infinite multiplicity for the eigenvalues. This degeneracy reflects the so called ``magnetic translation'' invariance for the system.
\item ii) Let $d_{{\cal{D}}_\xi}$ and ${\cal{A}}^\prime_{{\cal{D}}_\xi}$ denote respectively the spectral distance and the metric commutant for ${\cal{D}}_\xi$. From the first relation \eqref{lipnormgrossneveu}, one observes that $\xi=1$ implies ${\cal{A}}^\prime_{{\cal{D}}_\xi}={\cal{A}}$ so that the corresponding $d_{{\cal{D}}_\xi}$ cannot give rise to a spectral metric space. 
\end{itemize}
\end{remark}
In view of the remark \ref{remarklandau}, we now assume $\xi>0$. We further note that the infinite degeneracy of the spectrum for $H_L$ implies that the resolvent operator for ${\cal{D}}_\xi$ cannot be compact. This is different from the situation of the subsection \ref{harmonictriple} where the resolvent operator for the $D_k$'s are already compact. \par
We now derive some useful properties of integral kernels related to the Dirac operators considered in this subsection that will be needed to show the compactness property for the resolvent operators. In the following, the inverse operator for $H_L$ is defined as usual by $H_L^{-1}(x,y):=\int_0^\infty dt\ e^{-tH_L}(x,y):=\int_0^\infty dt\ {\cal{K}}_{H_L^{-1}}(t;x,y)$.\par 
\begin{proposition}\label{interkernel} One has the following properties:\hfill\break
i) For any $a\in{\cal{A}}$ and $\mu^2\in\mathbb{R}^+$, the integral kernel of the operator $L(a)(H_L+\mu^2)^{-1}$, ${\cal{K}}_{L(a)(H_L+\mu^2)^{-1}}$, is given by
\begin{equation}
{\cal{K}}_{L(a)(H_L+\mu^2)^{-1}}(x,y)={{e^{i2x\Theta^{-1}y}}\over{(\pi\theta)^2 }}\int d^2ud^2v\ a(u)\Phi(v)e^{i2x\Theta^{-1}(v-u)}e^{i2y\Theta^{-1}(u-\xi v)}e^{-i2u\Theta^{-1}v}\label{k1},
\end{equation}
\begin{equation}
\Phi(v)=\int_0^\infty dt{{e^{-t\mu^2}}\over{4\pi\sinh{t}}}e^{-\xi\theta^{-1}|v|^2\coth{t}}\label{fi1}.
\end{equation}
ii) For any $a\in{\cal{A}}$ and any $\xi>0$, $\xi\ne1$, $L(a)(H_L+\mu^2)^{-1}$ is a Hilbert-Schmidt operator on $L^2(\mathbb{R}^2)$.
\end{proposition}
\begin{proof}
First, one can write $(H_L+\mu^2)^{-1}=\int_0^\infty dt\ e^{-t\mu^2}e^{-tH_L}=\int_0^\infty dt\ e^{-t\mu^2}{\cal{K}}_{H_L^{-1}}(t;x,y)$. The determination of the integral kernel ${\cal{K}}_{H_L^{-1}}(x,y)$ can be obtained by a long but straightforward calculation by looking for a heat kernel solution of ${{d{\cal{K}}}\over{dt}}+H_L{\cal{K}}=0$ of the form $a(t)e^{b(t)|x-y|^2}e^{c(t)x\Theta^{-1}y}$, with $a(t),\ b(t),\ c(t)$ functions of $t$. It is given by
\begin{equation}
{\cal{K}}_{H_L^{-1}}(x,y)=\int_0^\infty dt\ {{1}\over{4\pi\sinh{t}}}e^{i2\xi x\Theta^{-1}y}e^{-\xi\theta^{-1}|x-y|^2\coth{t}}\label{khl}.
\end{equation}
Then, by combining \eqref{khl} with $K_{L(a)}$ \eqref{kla}, one obtains \eqref{k1}, \eqref{fi1} which hold for any $a\in{\cal{A}}$. This shows i).\par
To prove the property ii), one has to show that $I:=\int_{\mathbb{R}^2} dxdy|{\cal{K}}_{L(a)(H_L+\mu^2)^{-1}}(x,y)|^2$ is finite. By using \eqref{k1}, one obtains easily
\begin{equation}
I={{1}\over{(\pi\theta)^4 }}\int d^2xd^2yd^2u_1d^2u_2d^2v_1d^2v_2a^*(u_1)a(u_2)\Phi^*(v_1)\Phi(v_2)e^{i2u_1\Theta^{-1}v_1}e^{-i2u_2\Theta^{-1}v_2}\nonumber
\end{equation}
\begin{equation}
\times e^{-i2x\Theta^{-1}(v_1-v_2+u_2-u_1)}e^{-i4y\Theta^{-1}(u_1-u_2-\xi(v_1-v_2))}\nonumber
\end{equation}
\begin{equation}
={{1}\over{\pi^2\theta^2}}\int d^2yd^2u_2d^2v_1d^2v_2a^*(u_2+v_1-v_2)a(u_2)\Phi^*(v_1)\Phi(v_2)e^{i2(u_2-v_2)\Theta^{-1}v_1}e^{-i2u_2\Theta^{-1}v_2}\nonumber
\end{equation}
\begin{equation}
\times e^{-i4(1-\xi)y\Theta^{-1}(v_1-v_2)}\label{intermediaire}
\end{equation}
where the second equality is obtained upon integrating over $x$ and $u_1$ and we used the relation $\int {{d^2x}\over{(2\pi)^2}}e^{i2\alpha x\Theta^{-1}z}={{\theta^2}\over{4\alpha^2}}\delta^2(z)$, $\forall\alpha\in\mathbb{R}$, $\alpha\ne0$. When $\xi=1$, the last factor in the second equality equals to $1$ so that one obtains $I=+\infty$ and the operator $L(a)(H_L+\mu^2)^{-1}$ is not Hilbert-Schmidt on $L^2(\mathbb{R}^2)$. \\
When $\xi\ne1$, by further combining the delta function occurring from the integration over $y$ and the exponential factors appearing in \eqref{intermediaire}, one can express $I$ as
\begin{equation}
I=C\int d^2ud^2v |a(u)|^2|\Phi(v)|^2=C||a ||^2_2\int d^2v|\Phi(v)|^2\label{integral-estimate}
\end{equation}
where $C={{1}\over{4(1-\xi)^2 }}$. Now, by setting in \eqref{fi1} $u=\xi\theta^{-1}|v|^2(\coth{t}-1)$, one can rewrite $\Phi(v)$ as
\begin{equation}
\Phi(v)={{1}\over{4\pi}}\int_0^\infty dt\ e^{-\xi\theta^{-1}|v|^2}{{e^{-t}}\over{(t(t+2\xi\theta^{-1}|v|^2))^{{{1}\over{2}}} }}\big({{t}\over{ (t(t+2\xi\theta^{-1}|v|^2))^{{{1}\over{2}}} }} \big)^{\mu^2}
\end{equation}
from which one infers that
\begin{equation}
\Phi(v)\le{{1}\over{4\pi}}\int_0^\infty dt\ e^{-\xi\theta^{-1}|v|^2}{{e^{-t}}\over{(t(t+2\xi\theta^{-1}|v|^2))^{{{1}\over{2}}} }}\label{majorat}.
\end{equation}
The RHS of \eqref{majorat} can be exactly related to a modified Bessel function of second kind. This can be easily seen from the integral formula $K_0(xz)=({{\pi}\over{2z}})^{{{1\over{2}}}} {{e^{-xz}}\over{\Gamma({{1}\over{2}})}}\int_0^\infty dt\ e^{-tx}(t(1+{{t}\over{2z}}))^{-{{1}\over{2}}}$ which holds true whenever  $|arg(z)|<\pi,\ x>0$. Therefore, \eqref{majorat} translates into
\begin{equation}
\Phi(v)\le {{1}\over{4\pi}}K_0(\xi\theta^{-1}|v|^2),
\end{equation}
which, combined with $\int_0^\infty dxK_0(x)^2={{\pi^2}\over{4}}$ yields
\begin{equation}
\int d^2v\vert \Phi(v)\vert^2\le{{\pi\theta}\over{48\xi}}\label{estimation-2}.
\end{equation}
Then, combining \eqref{estimation-2} and \eqref{integral-estimate}, one obtains
\begin{equation}
I\le {\tilde{C}}||a ||^2_2<+\infty,
\end{equation}
with ${\tilde{C}}={{ \pi\theta}\over{192\xi(1-\xi)^2 }}$. This shows that $L(a)(H_L+\mu^2)^{-1}$ is a Hilbert-Schmidt operator on $L^2(\mathbb{R}^2)$ for any $a\in{\cal{A}}$, provided $\xi\ne1$. This terminates the proof.
\end{proof}
The property ii) of Proposition \ref{interkernel} implies:
\begin{corollary}\label{corol2}
For any $a\in{\cal{A}}$, $\xi>0$, $\xi\ne1$, $\mu\in\mathbb{R}^*$, $L(a)(H_L+\mu^2)^{-1}$ is a compact operator on $L^2(\mathbb{R}^2)$.
\end{corollary}
\begin{remark} \label{lszdual}A few comments are now in order:\par
\begin{itemize}
\item i) For $\xi=1$, one has $B={{2}\over{\theta}}={{1}\over{\theta_{LSZ}}}$ ($\theta_{LSZ}$ is the deformation parameter in the convention of \cite{LSZ}). Therefore, the value $\xi=1$ corresponds to the peculiar value for the magnetic field at which the LSZ model is exactly solvable \cite{LSZ}. However, the related noncommutative space modeled by $\mathfrak{Y}(1)$ is not a spectral metric space, in view of the point ii) of the remark \ref{remarklandau}. 
\item ii) Assume now that $\xi\ne1$. By observing directly the actual structure of eqns. \eqref{k1}, \eqref{fi1}, \eqref{khl} and the steps in the computation leading to the estimate of $I$, it can be easily realized that the fact that the operator $L(a)(H_L+\mu^2)^{-1}$ is Hilbert-Schmidt comes essentially from the occurrence of the left multiplication operator $L(a)$. Qualitatively speaking, this simply reflects the strong smoothing properties of the Moyal product. As it is almost apparent from \eqref{khl}, the operator $(H_L+\mu^2)^{-1}$ is obviously not Hilbert-Schmidt  which on a more physics inspirated viewpoint, reflects the infinite degeneracy of the eigenvalues of the Landau Hamiltonian operator on $\mathbb{R}^2$.\\
Notice that a somehow similar smoothing property of the left multiplication operator $L(a)$ has been exploited in \cite{GROSSEWULK2}, \cite{wulkgayral2010} where in particular the interesting dimension drop from (the KO-dimension) $2d$ to (the spectral dimension) $d$ in the corresponding harmonic spectral triples stems from the fact that  $|D_\Omega|^{-2d}$ and $L(a)|D_\Omega|^{-d}$ are {\it{both}} trace class operators.
\end{itemize}
\end{remark}
Thanks to Proposition \ref{interkernel}, Corollary \ref{corol2}, it can be easily realized that the following property is satisfied.
\begin{proposition}\label{Ytriple}
For any $\xi>0$, $\xi\ne1$, the triples ${\mathfrak{Y}}(\xi):=({\cal{A}},\ \pi_0,\ {\cal{H}}_0,\ {\cal{D}}_\xi)$, where $\pi_0:{\cal{A}}\to{\cal{B}}({\cal{H}}_0)$ is the left regular representation given by \eqref{regulrep}, are spectral triples with spectral distances $d_{{\cal{D}}_\xi}$ homothetic to $d_{D_0}$, namely
\begin{equation}
d_{{{\cal{D}}_\xi}}(\omega_1,\omega_2)=|1-\xi|^{-1}d_{D_0}(\omega_1,\omega_2),\ \forall\omega_1,\omega_2\in{\mathfrak{S}}({\cal{A}})\label{distanceDlandau}.
\end{equation}
Each ${\mathfrak{Y}}(\xi)$ describes a spectral metric space with infinite number of distinct connected components, each component being pathwise connected for the ${\cal{D}}_\xi$-topology.
\end{proposition}
\begin{proof}
It is straightforward to check that axioms i) and ii) of Definition \ref{spectraltriple} while axiom iii) is fullfilled thanks to the property \eqref{lipnormgrossneveu}. For any $a\in{\cal{A}}$, $\pi_0(a)$ has a diagonal action on ${\cal{H}}_0$ since one has $\pi_0(a)=L(a)\otimes\bbbone_2$. Besides, one can write
\begin{equation}
{\cal{D}}_\xi^2=\begin{pmatrix}
H_L-{{4\xi}\over{\theta}}&0\\
0& H_L+{{4\xi}\over{\theta}}
\end{pmatrix}.
\end{equation}
Let $\mu$ be a non zero real parameter. The Corollary \ref{corol2} implies that the diagonal operator $\pi_0(a)({\cal{D}}_\xi^2+\mu^2)^{-1}$ is a compact operator on ${\cal{H}}_0$ so that the axiom iv) of Definition \ref{spectraltriple} is verified. Hence for any $\xi>0$, $\xi\ne1$, ${\mathfrak{Y}}(\xi)$ is a spectral triple. The relation \eqref{distanceDlandau} is a direct consequence of the second relation \eqref{lipnormgrossneveu}.\\
The last part of the proposition can be shown in a way completely similar to the one used for the proposition \ref{connected1} which can be adapted to the present proposition. This terminates the proof.
\end{proof}
\begin{remark}\label{conseq2} 

The Proposition \ref{Ytriple} shows that \eqref{formula2} extends straightforwardly to:
\begin{equation}
d_{{\cal{D}}_\xi}(\omega_m,\omega_n)={\sqrt{{{\theta}\over{2(1-\xi)^2}} }}\sum_{p=n+1}^{m}{{1}\over{{\sqrt{p}} }},\ n<m. 
\end{equation}
and the conclusions obtained in \cite{mart-tom, mart-recent} apply to $\mathfrak{Y}(\xi)$ as seen before for $\mathfrak{X}(k)$, $k=1,2$.
\end{remark}
Consider the unital algebra ${\cal{M}}_\star:=({\cal{M}},\star)$. At a purely algebraic level, we note that ${\cal{D}}_\xi$ can be related to a connection on ${\cal{M}}_\star$ viewed as a module on itself. Indeed, as shown in \cite{WAL3, WAL4}, the map $\nabla_X:{\cal{M}}_\star\to{\cal{M}}_\star$, $X=(\partial_\mu),\ \mu=1,2$ given by
\begin{equation}
\nabla_{\partial_\mu}(a):=\nabla_\mu(a)=\partial_\mu a-iA_\mu\star a,\ \forall a\in{\cal{M}}_\star\label{nablamiu}
\end{equation}
where $A_\mu\in{\cal{M}}_\star$, $A_\mu^*=A_\mu$ defines a hermitian connection with hermitian structure
\begin{equation}
h:{\cal{M}}_\star\times{\cal{M}}_\star\to{\cal{M}}_\star,\ h(m_1,m_2)=m_1^*\star m_2, \forall m_1,m_2\in{\cal{M}}_\star\label{hermitstruct}.
\end{equation}
The unitary gauge group ${\cal{U}}({\cal{M}}_\star)$ acts on the connections as $\nabla_X^g=g^\dag\circ\nabla_X\circ g$ and yields
\begin{equation}
A_\mu^g:=g^\dag\star A_\mu\star g+g^\dag\star\partial_\mu g,\ \forall g\in{\cal{U}}({\cal{M}})\label{gaugetrans}.
\end{equation}
Whenever $A_\mu=-{{1}\over{2}}{\tilde{x}}_\mu:=A_\mu^{inv}$, the map \eqref{nablamiu} defines a gauge invariant connection \cite{WAL3, WAL4}:
\begin{equation}
\nabla^{inv}_\mu(a):=\partial_\mu a+{{i}\over{2}}{\tilde{x}}_\mu\star a={{i}\over{2}}a\star{\tilde{x}}_\mu,\ \forall a\in{\cal{M}}_\star\label{gaugeinvconnection},
\end{equation}
where gauge invariance $(\nabla^{inv}_X)^g=g^\dag\circ\nabla^{inv}_X\circ g=\nabla^{inv}_X$ follows from the second equality in the equation \eqref{gaugeinvconnection}. \\
Define now $\nabla^{\lambda}_\mu$ as $\nabla^{\lambda}_\mu(a):=\partial_\mu a+i\lambda{\tilde{x}}_\mu\star a$, $\forall a\in{\cal{M}}_\star$, with $A^{\lambda}_\mu:=-\lambda{\tilde{x}}_\mu$ with $\lambda\in\mathbb{R}$ and $\slashed{\nabla}^\lambda:=-i\sigma^\mu\nabla^\lambda_\mu$. By using Proposition \ref{formulas}, one obtains
\begin{equation}
{\cal{D}}_\xi=-i(1+\xi)\sigma^\mu\nabla^{{{\xi}\over{1+\xi}}}_\mu=(1+\xi)\slashed{\nabla}^{{{\xi}\over{1+\xi}}}\label{covderivation},
\end{equation}
which exhibits a relationship between the Dirac operator ${\cal{D}}_\xi$ and the ''covariant Dirac operator'' $\slashed{\nabla}^{{{\xi}\over{1+\xi}}}$ built from the connection $\nabla_\mu^\lambda$. From \eqref{distanceDlandau}, \eqref{covderivation} and \eqref{lipnormgrossneveu}, one obtains
\begin{equation}
||[\slashed{\nabla}^{{{\xi}\over{1+\xi}}},\pi_0(a)] ||={{1}\over{1+\xi}}||[{\cal{D}}_\xi,\pi_0(a)] ||={{|1-\xi|}\over{1+\xi}}||[D_0,\pi_0(a)] ||,\ \forall a\in{\cal{A}}.\label{lipnorms2}
\end{equation}
From \eqref{covderivation} and \eqref{lipnorms2}, it follows that, in some sense, determing the spectral distance for the spectral triples $\mathfrak{Y}(\xi)$ amounts to determine the spectral distance for the standard Dirac operator $D_0$ ``perturbed'' by a connection defined by ${\tilde{A}}^\xi_\mu=-{{\xi}\over{1+\xi}}{\tilde{x}}_\mu$. The resulting spectral metric spaces are homethetic thanks to the simple relation \eqref{distanceDlandau} (or \eqref{lipnorms2}). \\
Whenever $\xi=1$, ${\tilde{A}}^1_\mu=-{{1}\over{2}}{\tilde{x}}_\mu$ and the resulting connection is the gauge invariant connection. In particular, \eqref{covderivation} becomes ${\cal{D}}_{\xi=1}=2\slashed{\nabla}^{inv}$ which however does not give rise to a spectral metric space since the spectral distance between each pair of states in $\mathfrak{S}({\cal{A}})$ is infinite.\par 

It is easy to observe that a simple twist as in \eqref{twistdirac} of $D_k$, $k=1,2$, \eqref{Domega1}, \eqref{Domega2}, gives rise to Dirac operators involving two copies of \eqref{grossdirac}. Focusing on $D_1$, we quote directly the corresponding metric property:
\begin{proposition}\label{tripletcrossed}
For any $\xi>0$, $\xi\ne1$, the triples $({\cal{A}},\ \pi,\ {\cal{H}}=L^2(\mathbb{R}^2)\otimes\mathbb{C}^4,\ {\tilde{D}}_\xi)$, where $\pi:{\cal{A}}\to{\cal{B}}({\cal{H}})$ is the left regular representation \eqref{pi4} and the self-adjoint Dirac operator ${\tilde{D}}_\xi$, with Dom$({\tilde{D}}_\xi)={\cal{S}}\otimes\mathbb{C}^4$ is given by
\begin{equation}
{\tilde{D}}_\xi:=(\bbbone_2\otimes\sigma^\mu)(-i\partial_\mu)-\xi(\sigma^3\otimes\sigma^\mu)m({\tilde{x}}_\mu)\ \label{twistdirac}
\end{equation}
are spectral metric spaces having infinite number of distinct connected components, each pathwise connected, with corresponding spectral distances satisfying
\begin{equation}
d_{{\tilde{D}}_\xi}(\omega_1,\omega_2)={{1}\over{1+\xi}}d_{D_0}(\omega_1,\omega_2), \forall\omega_1,\omega_2\in\mathfrak{S}({\cal{A}}).\label{doubledist}
\end{equation}
\end{proposition}
\begin{proof}
Checking that axioms i) and ii) of definition \ref{spectraltriple} are satisfied is similar to what has been done for the spectral triples $(\mathfrak{X}(k))_{k=1,2}$. Next, one can write 
\begin{equation}
{\tilde{D}}_\xi=\begin{pmatrix}
-i\sigma^\mu(\partial_\mu-i\xi m({\tilde{x}}_\mu))&0\\
0&-i\sigma^\mu(\partial_\mu+i\xi m({\tilde{x}}_\mu))
\end{pmatrix}=
\begin{pmatrix}
{\cal{D}}_{-\xi}&0\\
0&{\cal{D}}_\xi
\end{pmatrix}
\end{equation}
which has therefore diagonal action on ${\cal{H}}={\cal{H}}_0\otimes\mathbb{C}^2$. It follows that
\begin{equation}
[{\tilde{D}}_\xi,\pi(a)]=\begin{pmatrix}
(1-\xi)[D_0,\pi_0(a)]&0\\
0&(1+\xi)[D_0,\pi_0(a)]
\end{pmatrix},\ \forall a\in{\cal{A}}.
\end{equation}
Therefore, one has for any $a\in{\cal{A}}$
\begin{equation}
\ell_{{\tilde{D}}_\xi}=||[{\tilde{D}}_\xi,\pi(a)] ||=\max(|1-\xi| ||[D_0,\pi_0(a)]||,|1+\xi| ||[D_0,\pi_0(a)]||)=(1+\xi)\ell_{D_0}(a).\label{homot2}
\end{equation}
Hence, $[{\tilde{D}}_\xi,\pi(a)]$ is a bounded operator on ${\cal{H}}$ so that axiom iii) of Definition \ref{spectraltriple} is satisfied while \eqref{doubledist} stems again from the definition of the spectral distance combined with \eqref{homot2}. \par 
Finally, to check the compactness axiom iv) on the resolvent operator of ${\tilde{D}}_\xi$, use for instance the fact that
\begin{equation}
 {\tilde{D}}_\xi^2=\begin{pmatrix}{\cal{D}}_{-\xi}^2&0\\ 
0&{\cal{D}}_{\xi}^2\end{pmatrix}
\end{equation}
together with Proposition \ref{interkernel}. It implies that $\pi(a)({\tilde{D}}^2_\xi+\mu^2)^{-1}\in{\cal{K}}({\cal{H}})$ so that the axiom iv) is satisfied. Finally, the existence of an infinite number of distinct pathwise connected components can be shown as for the Proposition \ref{connected1}.
\end{proof}
A similar discussion holds for $D_2$. 
\begin{remark}\label{finalremark} We close this section by noting that Proposition \ref{interkernel} (and Corollary \ref{corol2}) can be extended to any $\xi\in\mathbb{R}$, $\xi\ne1$. In this case, the expression for \eqref{k1} is unchanged except $\Phi(v)$ \eqref{fi1} which must be replaced by $\Psi(v)$ given by 
\begin{equation}
\Psi (v)=\int_0^\infty dt{{e^{-t\mu^2}}\over{4\pi\sinh{t}}}e^{-\vert\xi\vert\theta^{-1}|v|^2\coth{t}}\nonumber.
\end{equation}
Accordingly, $\int d^2v \vert\Psi (v) \vert^2\le{{\pi\theta }\over{48\vert\xi\vert }}$. Then, it follows that for any $a\in{\cal{A}}$, $\xi\in\mathbb{R}$, $\xi\ne1$, $L(a)(H_L+\mu^2)$ is Hilbert-Schmidt, hence compact on $L^2({\mathbb{R}^2})$. This implies that Proposition \ref{tripletcrossed} given above, together with Proposition \ref{Ytriple} are still valid for any $\xi\in\mathbb{R}$, $\xi\ne1$.
\end{remark}
\section{Discussion}
We now briefly discuss and comment the results.\par
Observe that \eqref{da} can be recast into the form $[D_\Omega,\pi(a)]=-iL(\partial_\mu a)\otimes\Gamma^\mu_{D_\Omega}$ for any $a\in{\cal{A}}$ where  $\Gamma^\mu_{D_\Omega}:=\gamma^\mu+\Omega\gamma^{\mu+2}$, $\mu=1,2$. Thanks to \eqref{cliff4}, the hermitian matrices $\Gamma^\mu_{D_\Omega}$ satisfy the Clifford relation 
\begin{equation}
\{\Gamma^\mu_{D_\Omega},\Gamma^\nu_{D_\Omega}\}=2\delta^{\mu\nu}(1+\Omega^2)\bbbone_4:=2(G_{D_\Omega}^{-1})^{\mu\nu}\bbbone_4\label{cliffmetrics}
\end{equation}
which defines an effective Clifford metric $G_{D_\Omega}$ for $D_\Omega$ through the rightmost side. A similar definition applies obviously for the Dirac operators $D_0$ and ${\cal{D}}_\xi$ in view of the equation \eqref{lipnormgrossneveu} and $[D_0,\pi_0(a)]=-iL(\partial_\mu a)\otimes \sigma^\mu$. Keeping the above observation in mind, we summarize the results of this paper by the following theorem:
\begin{theorem}\label{summaryresult}
The triples $({\cal{A}},\pi_0,{\cal{H}}_0,D_0)$ defined in proposition \ref{X0}, $({\cal{A}},\pi,{\cal{H}},D_k)$, $k=1,2$, defined in Theorem \ref{th1} and  $({\cal{A}},\pi_0,{\cal{H}}_0,{\cal{D}}_\xi)$ defined in proposition \ref{Ytriple} (and remark \ref{finalremark}) are homothetic spectral metric spaces, each having an infinite number of distinct components pathwise connected for the respective spectral distance. The homothety relations are given by
\begin{equation}
d_{D_k}(\omega_1,\omega_2)=(\det\ G_{D_k})^{{{1}\over{4}}}d_{D_0}(\omega_1,\omega_2), \ \forall k=1,2,\ \forall\omega_1,\omega_2\in{\mathfrak{S}}({\cal{A}}),\label{homott1}
\end{equation}
\begin{equation}
d_{{{\cal{D}}_\xi}}(\omega_1,\omega_2)=(\det\ G_{{{\cal{D}}_\xi}})^{{{1}\over{4}}}d_{D_0}(\omega_1,\omega_2),\ \forall\omega_1,\omega_2\in{\mathfrak{S}}({\cal{A}}),\label{homott2}
\end{equation}
where the effective Clifford metrics $G_{D_k}$ and $G_{{{\cal{D}}_\xi}}$ are defined as in \eqref{cliffmetrics}.
\end{theorem}
\begin{proof}
This follows from Proposition \ref{X0}, Theorem \ref{th1} and Proposition \ref{Ytriple}.
\end{proof}

It can be realized that the homothetic relations linking the spectral metric spaces occur thanks to the fact that the spatial derivatives can be expressed as inner derivatives (recall $\partial_\mu f=-{{i}\over{2}}[{\tilde{x}}_\mu,f]_\star$). This can be combined with the $x$-dependant term of each of the Dirac operators for the Harmonic and Landau case. As a result, one obtains $[D,\pi(a)]=\Gamma^\mu\pi(\partial_\mu a)$, where in each case the set of hermitian matrices $\Gamma^\mu$ spans a representation of a Clifford algebra.\par
As a remark from a more physics inspirated viewpoint, we note that the parameters $\Omega$ and $\xi$ in \eqref{Domega} and \eqref{grossdirac} should also appear in classical actions linked in some way with the corresponding spectral triples. These actions can be obtained basically in two ways: either from some algebraic construction, such as (most of) the one considered in \cite{GROSSWULK0}-\cite{WAL5}, or from a computation of the spectral action stemming from the associated spectral triple (supplemented by suitable additional conditions). In any case, it is natural to interpret these actions as field theories built on noncommutative space described by the corresponding spectral triple, whose metric properties are ruled by the spectral distance. Whenever the action is renormalisable, if these parameters are renormalised, then the spectral distance is affected by the corresponding renormalisation effects, in view of the homothety relations \eqref{homott1}, \eqref{homott2}. In some sense, metric properties of the noncommutative space would closely fit with the renormalisation of the field theory built on it. For instance, the parameter $\Omega$ has a renormalisation flow going to $1^-$ in the UV region within the renormalisable noncommutative $\varphi^4$ harmonic model \cite{GROSSWULK0,GROSSEWULK}, so that from \eqref{homott1} the spectral distance $d_{D_k}$ experienced by two states decreases when moving toward the UV region. A somewhat similar comment would apply for $d_{{\cal{D}}_\xi}$ viewed as the distance function on the noncommutative space underlying LSZ type models. Notice that $d_{{\cal{D}}_\xi}\to\infty$ when $\xi\to1^-$ so that, in that picture, (well defined) metric structure would show up only when $\xi$ renormalises away from $1$, (presumably) away from the UV regime.

\subsection*{Acknowledgments}
Discussions with P. Martinetti are gratefully acknowledged. We thank the referee for useful suggestions and comments. This work is dedicated to the memory of M.-C. Chantelouve.
\vskip 2 true cm
\appendix
\section{The Standard Moyal Spectral Metric Space}
In this appendix, we collect additional details relative to subsection 2.1. Let ${\cal{A}}:=({\cal{S}},\star)$ denotes the non-unital involutive algebra built from ${\cal{S}}$ equipped with the Moyal product \cite{GRACIAVAR}.
\begin{proposition}\cite{GAYRAL2004, GRACIAVAR}
${\cal{A}}$ is a non unital Fr\'echet pre-C* algebra equipped with usual seminorms $\rho_{\alpha,\beta}(a)=\sup_{x\in\mathbb{R}^2}|x_1^{\alpha_1}x_2^{\alpha_2}\partial_1^{\beta_1}\partial_2^{\beta^2}a|$, $\forall a\in{\cal{A}}$.
\end{proposition}
The $\star$-product \eqref{moyal} can be extended to spaces larger than ${\cal{S}}$. It is convenient to introduce the family of Wigner transition eigenfunctions of the harmonic oscillator $\{f_{mn}\}_{m,n\in\mathbb{N}}\subset{\cal{S}}$ (see for instance \cite{Gracia-Bondia:1987kw, Varilly:1988jk}), the so-called matrix base. This latter can be used to define a Frechet algebra isomorphism between ${\cal{A}}$ and the Frechet algebra of rapid decay matrices $(a_{mn})_{m,n\in\mathbb{N}}$ equipped with family of seminorms given by
\begin{equation}
 \rho^2_k(a)=\sum_{m,n\in\mathbb{N}}\theta^{2k}(m+{{1}\over{2}})^k(n+{{1}\over{2}})^k|a_{mn}|^2,\ k\in\mathbb{N}.
\end{equation}
It is given by 
\begin{equation}
(a_{mn})\to\sum_{m,n}a_{mn}f_{mn}\in{\cal{S}},
\end{equation}
with inverse
\begin{equation}
a\in{\cal{S}}\to a_{mn}={{1}\over{2\pi\theta}}\int d^2x(a\star f_{nm})(x)={{1}\over{2\pi\theta}}\int d^2x f_{mn}^\star(x) a(x).
\end{equation}
Then, the extension of the $\star$-product to various subspaces of ${\cal{S}}^\prime$ can be achieved by considering the family of spaces $\{{\cal{G}}_{s,t}\}_{s,t\in\mathbb{R}}$ defined as
\begin{equation}
{\cal{G}}_{s,t}:=\{ a=\sum_{m,n\in{\mathbb{N}}}a_{mn}f_{mn}\in{\cal{S}}^\prime;\ ||a ||_{s,t}^2=\sum_{m,n}\theta^{s+t}\big(m+\tfrac{1}{2}\big)^s\big(n+\tfrac{1}{2})^t|a_{mn}|^2<\infty\}, \label{gst}
\end{equation}
obtained by completion of the Schwarz algebra with respect to the norm $||.||_{s,t}$ given in \eqref{gst}. For any $a=\sum_{m,n}a_{mn}f_{mn}\in{\cal{G}}_{s,t}$ and any $b\in{\cal{G}}_{q,r}$, $b=\sum_{m,n}b_{mn}f_{mn}$, with $t+q\ge0$, the sequences $c_{mn}=\sum_{p}a_{mp}b_{pn}$, $\forall\, m,n\in{\mathbb{N}}$ define the functions $c=\sum_{m,n}c_{mn}f_{mn}$, $c\in{\cal{G}}_{s,r}$, as $||a\star b||_{s,r}\le||a||_{s,t}||b||_{q,r}$, $t+q\ge0$ and $|| a||_{u,v}\le ||a ||_{s,t}$ if $u\le s$ and $v\le t$. For more details, see e.g.~\cite{Gracia-Bondia:1987kw, Varilly:1988jk}. In particular, ${\cal{G}}_{0,0}=L^2(\mathbb{R}^2)$ and the dense and continuous inclusion ${\cal{S}}\subset{\cal{G}}_{s,t}\subset{\cal{S}}^\prime$ holds true for any $s,t\in\mathbb{R}$.\par 

Let ${\bar{{\cal{A}}}}$ be the C*-completion of ${\cal{A}}$. The subspace of pure states has been characterized e.g in \cite{WAL1}, Proposition 4, given below:
\begin{proposition}\label{purestates}
The pure states of ${\bar{ {\cal{A}}}}$ are the vector states $\omega_\psi:{\bar{ {\cal{A}}}}\to\mathbb{C}$ defined by any unit vector $\psi\in L^2(\mathbb{R}^2)$.
\end{proposition}
An explicit formula for the spectral distance between pure states related to the eigenfunctions of the harmonic oscillator has been constructed and studied in \cite{WAL1, WAL2}. For any $m\in\mathbb{N}$, we denote by $\omega_m$ the corresponding pure states which are generated by the unit vector ${{1}\over{{\sqrt{2\pi\theta}}}}f_{m0}\in L^2(\mathbb{R}^2)$. Let $d_{D_0}$ be the spectral distance in the notation of the present paper. One has the following spectral distance formula:
\begin{proposition}\label{distance-basic}\cite{WAL1}, \cite{WAL2}
The spectral distance between any two pure states $\omega_m$ and $\omega_n$ is
\begin{equation}
d_{D_0}(\omega_m,\omega_n)={\sqrt{{{\theta}\over{2}}}}\sum_{k=n+1}^m{{1}\over{{\sqrt{k}}}},\qquad \forall \, m,n\in{\mathbb{N}},\qquad n<m . \label{distancemn}
\end{equation}
\end{proposition}
\begin{proof}
The details of the proof can be found in e.g \cite{WAL1} (see Theorem 1).
\end{proof}
\begin{proposition}\cite{GAYRAL2004}
${\mathfrak{X}}_{D_0}=({\cal{A}},\ \pi_0,\ {\cal{H}}_0=L^2({\mathbb{R}}^2)\otimes {\mathbb{C}}^2,\ D_0)$ where $\pi_0:{\cal{A}}\to{\cal{B}}({\cal{H}}_0)$ is the left regular representation \eqref{regulrep} and $D_0$ is the standard Dirac operator on $\mathbb{R}^2$ given in \eqref{cliff2} is a non unital spectral triple as in Definition \ref{spectraltriple}.
\end{proposition}
\begin{proof}
The axioms i) and ii) of Definition \ref{spectraltriple} are satisfied by construction, owing to the properties of the Moyal product, ${\cal{A}}$ and $D_0$ and the density of ${\cal{S}}$ in $L^2(\mathbb{R}^2)$ from which follows the density of Dom$(D_0)$ in ${\cal{H}}_0$. \\
Furthermore, the following estimate holds true: $||a\star b||_2\le{{1}\over{{\sqrt{2\pi\theta}}}}||a ||_2||b||_2$, $\forall a,b\in L^2(\mathbb{R}^2)$. This can be straightforwardly obtained by expanding $a$ and $b$ in the matrix basis $\{f_{mn}\}_{m,n\in\mathbb{N}}$ used as an orthonormal basis of $L^2(\mathbb{R}^2)$, \cite{Gracia-Bondia:1987kw, Varilly:1988jk}, and making use of the Cauchy-Schwarz inequality so that one readily obtains that $L(a)$ is a bounded operator and therefore $\pi_0(a)\in{\cal{B}}({\cal{H}}_0)$ with $\pi_0$ given in \eqref{regulrep}. Then, Proposition \ref{lzero} implies $[D_0,\pi_0(a)]\in{\cal{B}}({\cal{H}}_0)$ for any $a\in{\cal{A}}$. Hence, the axiom iii) is satisfied.\\
To prove that iv) is satisfied, we make use of a textbook property (see e.g \cite{GRACIAVAR}) that for any self adjoint operator $D$, $\pi(a){{1}\over{D-\lambda}}\in {\cal{K}}({H})$, for any $\lambda\notin$ sp$D$ $\iff $ $\pi(a){{1}\over{D^2+1}}\in{\cal{K}}(H)$. Here, $\pi$ is assumed to be a $\star$-representation on some involutive algebra on the bounded operator of some Hilbert space $H$. This applies to $\mathfrak{X}_{D_0}$ given in Proposition \ref{X0}, so that one has to prove that for any $a\in{\cal{A}}$, $\pi_0(a){{1}\over{D_0^2+1}}\in{\cal{K}}({\cal{H}}_0)$, with $D_0$ given in \eqref{cliff2} and $\pi_0$ defined in \eqref{regulrep}. Then, by simply using $D_0^2=-\partial_\mu\partial^\mu\otimes\bbbone_2:=-\partial^2\otimes\bbbone_2$ and \eqref{regulrep}, one can write
\begin{equation}
\pi_0(a){{1}\over{D_0^2+1}}=L(a){{1}\over{-\partial^2+1}}\otimes\bbbone_2\label{diagresolv},
\end{equation}
which therefore acts diagonally on ${\cal{H}}_0=L^2(\mathbb{R}^2)\otimes\mathbb{C}^2$, so that it is sufficient to show that $L(a){{1}\over{-\partial^2+1}}$ for any $a\in{\cal{A}}$ is a compact operator on $L^2(\mathbb{R}^2)$. For that purpose, it is very convenient in the present situation to consider the corresponding integral kernel $K_{L(a)(-\partial^2+1)^{-1}}(x,y)$. Using the expression of the Moyal product given in Proposition \ref{moyalproperty}, one easily obtains the integral kernel for the bounded operator $L(a)$ given by
\begin{equation}
K_{L(a)}(x,y)={{1}\over{( \pi\theta)^2}}\int d^2z\ a(x+z)e^{i\,2z^\mu\,\Theta^{-1}_{\mu\nu}(x^\nu-y^\nu)},\ \forall a\in{\cal{A}}\label{kla}
\end{equation}
from which follows ($C$ is a real constant)
\begin{equation}
K_{L(a)(-\partial^2+1)^{-1}}(x,y)=C\int d^2p\ {{a(x+{{1}\over{2}}\Theta p)}\over{p^2+1}}e^{ip(x-y)},\ \forall a\in{\cal{A}}\label{kernelresolvante1}.
\end{equation}
Consider now $I:=\int d^2xd^2y|K_{L(a)(-\partial^2+1)^{-1}}(x,y)|^2$. Using \eqref{kernelresolvante1}, one can write
\begin{align}
I&=C^2\int d^2xd^2yd^2p_1d^2p_2\ {{a^*(x+{{1}\over{2}}\Theta p_1)}\over{p_1^2+1}}
{{a(x+{{1}\over{2}}\Theta p_2)}\over{p_2^2+1}}e^{-ip_1(x-y)}e^{ip_2(x-y)} \nonumber\\
&=C^2\int d^2xd^2p\ {{a^*(x+{{1}\over{2}}\Theta p)}\over{p^2+1}}
{{a(x+{{1}\over{2}}\Theta p)}\over{p^2+1}}=C^2\int d^2xd^2p\ |a(x+{{1}\over{2}}\Theta p)|^2({{1}\over{p^2+1}})^2\nonumber\\
&=(C^\prime)^2\ ||a||^2_2\ ||{{1}\over{p^2+1}}||^2_2<+\infty,\ \forall a\in{\cal{A}},\label{HS1}
\end{align}
where the last equality is obtained through a change of variable. Then, \eqref{HS1} implies that the operator $L(a)(-\partial^2+1)^{-1}$ is a Hilbert-Schmidt operator on $L^2(\mathbb{R}^2)$. Therefore, it is compact on $L^2(\mathbb{R}^2)$, so that, in view of \eqref{diagresolv}, $\pi_0(a){{1}\over{D_0^2+1}}\in{\cal{K}}({\cal{H}}_0)$ for any $a\in{\cal{A}}$ holds true and axiom iv) is satisfied. This terminates the proof.
\end{proof}


\begin{thebibliography}{50} 

\bibitem{CONNES} A. Connes, ``Noncommutative Geometry,'' Academic Press Inc., San Diego 
(1994), available at {http://www.alainconnes.org/downloads.html}

\bibitem{CM} 
A. Connes and M. Marcolli, ``A walk in the noncommutative garden,'' 
(2006), available at {http://www.alainconnes.org/downloads.html}

\bibitem{LANDI} 
G. Landi, ``An introduction to noncommutative spaces and their geometries'', Lectures notes in physics, Springer-Verlag (Berlin, Heidelberg) (1997).

\bibitem{GRACIAVAR} 
J.~M. Gracia-Bond{\'\i}a, J.~C. V{\'a}rilly and H. Figueroa, ``Elements of Noncommutative Geometry", Birkha\"user Advanced Texts, Birkha\"user Boston, Basel, Berlin (2001).

\bibitem{CONNES1} A. Connes, ''Compact metric spaces,  Fredholm modules and hyperfiniteness``, Ergodic Theory Dynam. Systems {\bf{9}} (1989) 207-220.

\bibitem{CONNES2} A. Connes, ''Gravity coupled with matter and the foundation of Noncommutative geometry``, Commun. Math. Phys. {\bf{182}} (1996) 155-176. 

\bibitem{Rieffel1} M. A. Rieffel, ``Metrics on state spaces'', Doc. Math. {\bf{4}} (1999) 559-600.

\bibitem{Rieffel11} M. A. Rieffel, ''Metrics on states from actions of compact groups'', Doc. Math. {\bf{3}} (1998) 215-229.

\bibitem{Rieffel2a} M.A. Rieffel, ``Compact quantum metric spaces'', in ``Operator algebras, quantization and noncommutative geometry'', Contemp. Math. {\bf{365}} 315-330.

\bibitem{Rieffel2b} M.A. Rieffel, ``Gromov-Hausdorff distance for quantum metric spaces. Matrix algebras converge to the sphere for quantum Gromov-Hausdorff distance'', Mem. Amer. Math. Soc. {\bf{168}} (2004), n$^o$ 796, 1-65.

\bibitem{Rieffel3} M.A. Rieffel, ''Group C*-algebras as compact quantum metric spaces '', Doc. Math. {\bf{7}} (2002) 605-651. 

\bibitem{Latremolieres} F. Latr\'emoli\`ere, ``Bounded-Lipschitz distances on the state space of a C*-algebra'', Taiwanese J. of Math. {\bf{11}} (2007) 447-469.

\bibitem{Muller1} A. Dimakis, F. M\"uller-Hoissen and T. Striker, ``Non-commutative differential calculus and lattice gauge theories'', J. Phys. A: Math. Gen. {\bf{26}} (1993) 1927-1949. 

\bibitem{Muller2} A. Dimakis, F. M\"uller-Hoissen, `` Conne's distance function on one dimensional lattices'', Int. J. Theor. Phys. {\bf{37}} (1998) 907-913.

\bibitem{LIZZI} G. Bimonte, F. Lizzi and G. Sparano, ``Distances on a Lattice from Non-Commutative Geometry'', Phys. Lett. B{\bf{341}} (1994) 139-146.

\bibitem{Mart1} B. Iochum, T. Krajewski and P. Martinetti, ``Distances in finite spaces from noncommutative geometry'', J. Geom. Phys. {\bf{37}} (2001) 100-125.

\bibitem{Mart2} P. Martinetti, ``Carnot-Caratheodory metric and gauge fluctuation in noncommutative geometry'', Commun. Math. Phys. {\bf{265}} (2006) 585-616.

\bibitem{Martwulk} P. Martinetti and R. Wulkenhaar, ``Discrete Kaluza-Klein from scalar fluctuations in noncommutative geometry'', J. Math. Phys. {\bf{43}} (2002) 182-204.

\bibitem{Mart3} P. Martinetti, ``Spectral distance on the circle'', J. Funct. Ana. {\bf{255}} (2008) 1575-1612.

\bibitem{WAL1} E. Cagnache, J.-C. Wallet, ''Spectral distances: Results for Moyal plane and noncommutative torus'', SIGMA {\bf{6}} (2010) 026.

\bibitem{WAL2} E. Cagnache, F. d'Andrea, P. Martinetti and J-C. Wallet, ``The spectral distance on the Moyal plane'', J.Geom.Phys.{\bf{61}} (2011) 1881-1897.

\bibitem{GAYRAL2004}
V. Gayral, J.~M. Gracia-Bond{\'\i}a, B. Iochum, T. Sch\"ucker and J.~C.~V{\'a}rilly, ``Moyal planes are spectral triples", Commun.\ Math.\ Phys.\  {\bf 246} (2004) 569-623.

\bibitem{MARSE1} V. Gayral and B. Iochum, ''The spectral action for Moyal planes``, J. Math. Phys. {\bf{46}} (2005) 043503.

\bibitem{DFR} S. Doplicher, K. Fredenhagen, and J. E. Robert, ``The quantum structure of spacetime at the
Planck scale and quantum fields'',  Commun.Math.Phys. {\bf{172}} (1995) 187-220.

\bibitem{mart-tom} P. Martinetti, F. Mercati and L. Tomassini, ``Minimal length in quantum space and integrations of the line element in Noncommutative Geometry'', Rev. Math. Phys. {\bf{24}} (2012) 1250010.

\bibitem{mart-recent} P. Martinetti, L. Tomassini, ``Noncommutative geometry of the Moyal plane: translation isometries, Connes's spectral distance between coherent states, Pythagoras equality`` arXiv:1110.6164 (2011).

\bibitem{GROSSWULK0} H. Grosse and R. Wulkenhaar, ''Renormalisation of $\varphi^4$-theory on noncommutative $\mathbb{R}^2$ in the matrix base'', JHEP {\bf{0312}} (2003) 019.

\bibitem{GROSSEWULK} H. Grosse and R. Wulkenhaar, ``Renormalisation of the $\varphi^4$-theory on noncommutative $\mathbb{R}^4$ in the matrix base'', Commun. Math. Phys. {\bf{256}} (2005) 305-374.

\bibitem{VIGNES1} F. Vignes-Tourneret,''Renormalization of the Orientable Non-commutative Gross-Neveu Model``, Ann. H. Poincar\'e {\bf{8}} (2007) 427-474.

\bibitem{VIGNEWALLET} A. Lakhoua, F. Vignes-Tourneret and J.-C. Wallet, ''One-loop Beta Functions for the Orientable Non-commutative Gross-Neveu Model``, Eur. Phys. J. C{\bf{52}} (2007) 735-742. 

\bibitem{rivass} M. Disertori, R. Gurau, J. Magnen, V. Rivasseau, ''Vanishing of the Beta function of non commutative $\Phi^4_4$ theory to all orders'', Phys. Lett. B{\bf{649}} (2007) 95-102.

\bibitem{LSZ} E. Langmann, R. J. Szabo and K. Zarembo, ''Exact solution of quantum field theory on noncommutative phase spaces``, JHEP {\bf{0401}} (2004) 017.

\bibitem{WAL6} A. de Goursac, J-C. Wallet and R. Wulkenhaar, ``Noncommutative induced gauge theory'', Eur. Phys. J. C{\bf{51}} (2007) 977-987.

\bibitem{GROSSEWOHL} H. Grosse and M. Wohlgenannt, ``Induced Gauge Theory on a Noncommutative space'', Eur. Phys. J. C{\bf{52}} (2007) 435-450.

\bibitem{vacuum} A. de Goursac, J.-C. Wallet, R. Wulkenhaar, ''On the vacuum states for noncommutative gauge theory'', Eur. Phys. J. C{\bf{56}} (2008) 293-304.

\bibitem{WAL3} J-C. Wallet, ``Derivations of the Moyal algebra and Noncommutative gauge theories'', SIGMA {\bf{5}} (2009) 013.

\bibitem{WAL4} E. Cagnache, T. Masson and J-C. Wallet, ''Noncommutative Yang-Mills-Higgs actions from derivation based differential calculus``, 
J. Noncommut. Geom. {\bf{5}} (2011) 39-67.

\bibitem{WAL5} A. de Goursac, T. Masson, J-C. Wallet, ''Noncommutative $\varepsilon$-graded connections and application to Moyal space'', J. Noncommut. Geom. {\bf{6}} (2012) 343-387.

\bibitem{GROSSEWULK2} H. Grosse and R. Wulkenhaar, "8-D spectral triple on 4D-Moyal space and the vacuum of noncommutative gauge theory``, arXiv:0709.0095 [hep-th] (2007).

\bibitem{wulkgayral2010} V. Gayral and R. Wulkenhaar, ''Spectral geometry of the Moyal plane with harmonic propagation``, arxiv: 1108.2184 (2011). 

\bibitem{wulk-finite} R. Wulkenhaar, ''Non-compact spectral triples with finite volume``, in {\it{Quanta of Maths.}}, Clay Math. Proc. {\bf{11}} (2010) 617, AMS, Providence, RI.

\bibitem{bellissardmarcolli} J. V. Bellissard, M. Marcolli and K. Reihani, "Dynamical systems on spectral metric spaces``, arXiv:1008.4617 [math:OA] (2010).


\bibitem{Gracia-Bondia:1987kw} 
J.~M. Gracia-Bond{\'\i}a and J.~C.
  V{\'a}rilly, ``Algebras of distributions suitable for phase space quantum mechanics. {I},'' J.\ Math.\ Phys.\ {\bf 29} (1988) 869-879.

\bibitem{Varilly:1988jk}
  J.~C.~V{\'a}rilly and J.~M.~Gracia-Bond{\'\i}a,
  ``Algebras of distributions suitable for phase-space quantum mechanics. {II}.
  Topologies on the Moyal algebra,'' J.\ Math.\ Phys.\  {\bf 29} (1988) 880-887.

\bibitem{Andreamartinetti} F. d'Andrea and P. Martinetti, ``A view on transport theory from noncommutative geometry'', SIGMA {\bf{6}} (2010) 057.

\bibitem{Rieffel5} M.A. Rieffel, ``Operator algebras, quantization and noncommutative geometry'', Contemp. Math {\bf{365}}, 315-330, (AMS, Providence, RI, 2004).

\bibitem{equadiff} M. Renardy, R.C. Rogers, ''An introduction to partial differential equations'', Texts in applied mathematics {\bf{13}}, Springer (2004).



\end{thebibliography}
\end{document}